\RequirePackage{fix-cm}
\documentclass[draft,smallextended]{svjour3}       
\smartqed  
\usepackage{graphicx,mathptmx}
\usepackage[utf8]{inputenc}
\usepackage[ngerman,english]{babel}
\usepackage{amsmath,xcolor,amsfonts,amssymb,dsfont,centernot}
\usepackage{stmaryrd,mathtools,tikz,tikz-qtree}
\usetikzlibrary{shapes,decorations,shadows,arrows,calc}
\tikzstyle{arg}=[draw,circle,fill=gray!15,inner sep=1pt,minimum size=.5cm]
\usepackage{forest,graphicx}
\usetikzlibrary{arrows,positioning,decorations.pathreplacing}

\allowdisplaybreaks

\useshorthands{"}
\addto\extrasenglish{\languageshorthands{ngerman}}
\DeclareRobustCommand{\nmodels}{\mathrel{|}\joinrel\mathrel{\text{\reflectbox{$\not\stackrel\forall=$}}}}

\DeclareMathOperator{\wt}{wt}
\DeclareMathOperator{\height}{ht}
\DeclareMathOperator{\pos}{pos}

\DeclareMathOperator{\supp}{supp}
\DeclareMathOperator{\yield}{yield}

\DeclareMathOperator{\var}{var}
\DeclareMathOperator{\ran}{ran}
\DeclareMathOperator{\rk}{rk}
\providecommand*{\nat}[0]{\ensuremath{\mathbb N}}
\providecommand*{\seq}[3]{\ensuremath{#1_{#2}, \dotsc, #1_{#3}}}
\providecommand*{\abs}[1]{\ensuremath{\lvert #1 \rvert}}
\providecommand*{\word}[3]{\ensuremath{#1_{#2} \dotsm #1_{#3}}}

\newcommand{\Tsigma}{T_{\Sigma}}

\newcommand{\N}{\mathbb{N}}

\journalname{Theory of Computing Systems}

\begin{document}

\title{Weighted Tree Automata with Constraints\thanks{This is an
    extended and revised version of [Maletti, N\'asz: Weighted Tree
    Automata with Constraints.  Proc.\@ 26th DLT, LNCS~13257,
    Springer~2022].}}
\author{Andreas Maletti
  \and
  Andreea-Teodora Nász}
\authorrunning{A.~Maletti and T.~N\'asz}
\institute{A. Maletti and T.~N\'asz \at
  Universit\"at Leipzig, Faculty of Mathematics and Computer Science \\
  PO~box 100\,920, 04009 Leipzig, Germany \\
  \email{$\{$maletti,nasz$\}$@informatik.uni-leipzig.de}}
\date{Received: December 7, 2022 / Accepted: ??}

\maketitle

\begin{abstract}
  The HOM problem, which asks whether the image of a regular tree
  language under a given tree homomorphism is again regular, is known
  to be decidable [Godoy \& Gim{\'e}nez: The HOM problem is decidable.
  JACM~60(4), 2013].  However, the problem remains open for regular
  weighted tree languages.  It is demonstrated that the main notion
  used in the unweighted setting, the tree automaton with equality and
  inequality constraints, can straightforwardly be generalized to the
  weighted setting and can represent the image of any regular weighted
  tree language under any nondeleting and nonerasing tree
  homomorphism.  Several closure properties as well as decision
  problems are also investigated for the weighted tree languages
  generated by weighted tree automata with constraints.
\end{abstract}

\section{Introduction}
\label{intro}
Numerous extensions of nondeterministic finite-state string automata
have been proposed in the past few decades.  On the one hand, the
qualitative evaluation of inputs was extended to a quantitative
evaluation in the weighted automata of~\cite{schutzenberger1961}.
This development led to the fruitful study of recognizable formal
power series~\cite{salomaa2012automata}, which are well-suited for
representing factors such as costs, consumption of resources, or time
and probabilities related to the processed input.  The main algebraic
structure for the weight calculations are
semirings~\cite{gol99,hebwei98}, which offer a nice compromise between
generality and efficiency of computation (due to their
distributivity).  On the other hand, finite-state automata have been
generalized to other input structures such as infinite
words~\cite{infinitewords} and trees~\cite{tataok}.  Finite-state tree
automata were introduced independently
in~\cite{doner1970tree,thatcher1965generalized,thatcher1968generalized},
and they and the tree languages they generate, called regular tree
languages, have been intensively studied since their
inception~\cite{tataok}.  They are successfully utilized in various
applications in many diverse areas like natural language
processing~\cite{jurmar08}, picture
generation~\cite{drewes2006grammatical}, and compiler
construction~\cite{wilseihac13}.  Indeed several applications require
the combination of the two mentioned generalizations, and a broad
range of weighted tree automaton~(WTA) models has been studied (see
\cite[Chapter~9]{fulvog09} for an overview).

It is well-known that finite-state tree automata cannot ensure that
two subtrees (of potentially arbitrary size) are always equal in an
accepted tree~\cite{gecste15}.  An extension proposed
in~\cite{rateg1981} aims to remedy this problem and introduces a tree
automaton model that explicitly can require certain subtrees to be
equal or different.  Such models are very useful when investigating
(tree) transformation models (see~\cite{fulvog09} for an overview)
that can copy subtrees (thus resulting in equal subtrees in the
output), and they are the main tool used in the seminal
paper~\cite{godoy2013hom} that proved that the HOM problem is
decidable.  The HOM problem was a long-standing open problem in the
theory of tree languages and recently solved in~\cite{godoy2013hom}.
It asks whether the image of an (effectively presented) regular tree
language under a given tree homomorphism is again regular.  This is
not necessarily the case as tree homomorphisms can create copies of
subtrees.  Indeed removing this ability from the tree homomorphism,
obtaining a linear tree homomorphism, yields that the mentioned image
is always regular~\cite{gecste15}.  In the solution to the HOM problem
provided in~\cite{godoy2013hom} the image is first represented by a
tree automaton with constraints, and then it is investigated whether
this tree automaton actually generates a regular tree language.

The HOM problem is also interesting in the weighted setting as it once
again provides an answer whether a given homomorphic image of a
regular weighted tree language can be represented efficiently.  While
preservation of regularity has been
investigated~\cite{bozrah05,esikui03,fulmalvog10,fulmalvog10b} also in
the weighted setting, the decidability of the HOM problem remains wide
open.  With the goal of investigating this problem, we introduce
weighted tree grammars with constraints (WTGc for short) in this
contribution.  We demonstrate that those WTGc can again represent all
(nondeleting and nonerasing) homomorphic images of the regular
weighted tree languages.  Thus, in principle, it only remains to
provide a decision procedure for determining whether a given WTGc
generates a regular weighted tree language.  We approach this task by
providing some common closure properties following essentially the
steps also taken in~\cite{godoy2013hom}.  For zero-sum free semirings
we can also show that decidability of support emptiness and finiteness
are directly inherited from the unweighted case~\cite{godoy2013hom}.

The present work is a revised and extended version
of~\cite{maletti2022weighted} presented at the 26th Int.\@
Conf.\@ Developments in Language Theory~(DLT 2022).  We 
provide additional proof details and examples, 
as well as a new pumping lemma for the class of (nondeleting and
nonerasing) homomorphic images of regular weighted tree languages.  We
utilize this pumping lemma to show that for any zero-sum free
semiring, the class of homomorphic images of regular weighted tree
languages is properly contained in the class of weighted tree
languages generated by all positive WTGc, which are WTGc that utilize
only equality constraints. 

\section{Preliminaries}
We denote the set of nonnegative integers by~$\N$, and we let $[k] =
\{i \in \N \mid 1 \leq i \leq k\}$ for every~$k \in \N$.  For all 
sets~$T$~and~$Z$ let~$T^Z$ be the set of all mappings~$\varphi \colon
Z \to T$, and correspondingly we sometimes write~$\varphi_z$ instead
of~$\varphi(z)$ for every~$\varphi \in T^Z$.  The inverse
image~$\varphi^{-1}(S)$ of~$\varphi$ for a subset~$S \subseteq T$
is~$\varphi^{-1}(S) = \{z \in Z \mid \varphi(z) \in S\}$, and we
write~$\varphi^{-1}(t)$ instead of~$\varphi^{-1}(\{t\})$ for every~$t
\in T$.  The \emph{range} of~$\varphi$ is
\[ \ran(\varphi) = \bigl\{\varphi(z) \mid z \in Z \bigr\} \enspace. \]
Finally, the cardinality of~$Z$ is denoted by~$\abs Z$.

A \emph{ranked alphabet}~$(\Sigma, \mathord{\rk})$ is a pair
consisting of a finite set~$\Sigma$ and a map~$\mathord{\rk} \in
\N^\Sigma$ that assigns a rank to each symbol of~$\Sigma$.  If there
is no risk of confusion, we denote a ranked alphabet~$(\Sigma,
\mathord{\rk})$ by~$\Sigma$.  We write~$\sigma^{(k)}$ to indicate
that~$\rk(\sigma) = k$.  Moreover, for every~$k \in \N$ we
let~$\Sigma_k = \rk^{-1}(k)$.  Let $X = \{x_i \mid i \in \N\}$ be a
countable set of (formal) variables.  For each~$k \in \N$ we let~$X_k
= \bigl\{x_i \mid i \in [k] \bigr\}$.  Given a ranked alphabet~$\Sigma$
and a set~$Z$, the set~$\Tsigma(Z)$ of \emph{$\Sigma$"~trees indexed
  by~$Z$} is the smallest set such that~$Z \subseteq \Tsigma(Z)$ and
$\sigma(\seq t1k) \in \Tsigma(Z)$ for every~$k \in \N$, $\sigma \in
\Sigma_k$, and~$\seq t1k \in \Tsigma(Z)$.  We
abbreviate~$\Tsigma(\emptyset)$ simply to~$\Tsigma$, and any subset~$L
\subseteq \Tsigma$ is called a \emph{tree language}.

Let $\Sigma$~be a ranked alphabet, $Z$~a set, and~$t \in\Tsigma(Z)$.
The set~$\pos(t)$ of \emph{positions of~$t$} is inductively defined
by~$\pos(z) = \{\varepsilon\}$ for all~$z \in Z$ and by
\[ \pos\bigl(\sigma(\seq t1k) \bigr) = \bigl\{\varepsilon \bigr\} \cup
  \bigcup_{i \in [k]} \bigl\{iw \mid w \in \pos(t_i) \bigr\} \] for
all~$k \in \N$, $\sigma \in \Sigma_k$, and~$\seq t1k \in \Tsigma(Z)$.
The size~$\abs t$ of~$t$ is defined as~$\abs t = \abs{\pos(t)}$, and
its height~$\height(t)$ is~$\height(t) = \max_{w \in \pos(t)} \abs w$.
For~$w\in \pos(t)$ and~$t' \in \Tsigma(Z)$, the \emph{label}~$t(w)$
of~$t$ at~$w$, the \emph{subtree}~$t|_w$ of~$t$ at~$w$, and the
\emph{substitution}~$t[t']_w$ of~$t'$ into~$t$ at~$w$ are defined
by~$z(\varepsilon) = z|_\varepsilon = z$ and~$z[t']_\varepsilon = t'$
for all~$z \in Z$ and for~$t = \sigma(\seq t1k)$
by~$t(\varepsilon) = \sigma$, $t(iw') = t_i(w')$,
$t|_\varepsilon = t$, $t|_{iw'} = t_i|_{w'}$,
$t[t']_\varepsilon = t'$, and
\[ t[t']_{iw'} = \sigma \bigl(\seq t1{i-1}, t_i[t']_{w'}, \seq t{i+1}k
  \bigr) \] for all~$k \in \N$, $\sigma \in \Sigma_k$,
$\seq t1k \in\Tsigma(Z)$, $i \in [k]$, and~$w' \in \pos(t_i)$.  For
all~$S \subseteq \Sigma \cup Z$, we
let~$\pos_S(t) = \bigl\{w \in \pos(t) \mid t(w) \in S \bigr\}$ and
$\var(t) = \bigl\{x \in X \mid \pos_x(t) \neq \emptyset \bigr\}$.  For
a single~$\sigma \in \Sigma \cup Z$ we
abbreviate~$\pos_{\{\sigma\}}(t)$ simply by~$\pos_\sigma(t)$.

The yield mapping~$\mathord{\yield} \colon T_\Sigma(Z) \to Z^*$ is
recursively defined by
\[ \yield \bigl(z \bigr) = z \qquad \text{and} \qquad \yield \bigl(
  \sigma(\seq t1k) \bigr) = \yield(t_1) \dotsm \yield(t_k) \] for
every~$z \in Z$, $k \in \nat$, $\sigma \in \Sigma_k$, and
trees~$\seq t1k \in T_\Sigma(Z)$.  A tree~$t \in T_\Sigma(Z)$ is
called~\emph{context} if~$\abs{\pos_z(t)} = 1$ for every~$z \in Z$.
We write~$C_\Sigma(Z)$ for the set of such contexts and
$\widehat C_\Sigma(X_k) = \bigl\{c \in C_\Sigma(X_k) \mid \yield(c) =
\word x1k \bigr\}$.  Finally, for every~$t\in \Tsigma(Z)$,
finite~$V \subseteq Z$, and~$\theta \in \Tsigma(Z)^V$, the
substitution~$\theta$ applied to~$t$ is written as~$t\theta$ and
defined by~$v\theta = \theta_v$ for every~$v \in V$, $z\theta = z$ for
every~$z \in Z \setminus V$, and
\[ \sigma(\seq t1k)\theta = \sigma(t_1\theta, \dotsc, t_k\theta) \]
for all~$k \in \N$, $\sigma \in \Sigma_k$,
and~$\seq t1k \in \Tsigma(Z)$.  We also write the
substitution~$\theta \in \Tsigma(Z)^V$ as $[v_1 \gets \theta_{v_1},
\dotsc, v_n \gets \theta_{v_n}]$ if~$V = \{\seq v1n\}$.  Finally, we
abbreviate it further to just~$[\theta_{v_1}, \dotsc, \theta_{v_n}]$
if~$V = X_n$.

A \emph{commutative semiring}~\cite{hebwei98,gol99} is a
tuple~$(\mathbb{S}, \mathord+, \mathord\cdot, 0, 1)$ such that
$(\mathbb{S}, \mathord+, 0)$~and $(\mathbb{S}, \mathord\cdot, 1)$ are
commutative monoids, $\cdot$~distributes over~$+$, and~$0 \cdot s = 0$ 
for all~$s\in \mathbb{S}$.  Examples include (i)~the Boolean
semiring~$\mathbb{B} = \bigl(\{ 0,1 \}, \mathord\vee, \mathord\wedge,
0, 1 \bigr)$, (ii)~the semiring~$\N = \bigl(\N , \mathord+,
\mathord\cdot, 0, 1)$, (iii)~the tropical semiring $\mathbb T =
\bigl(\N \cup\{\infty\}, \mathord{\min}, \mathord+, \infty, 0 \bigr)$,
and (iv)~the arctic semiring $\mathbb A = \bigl(\N \cup\{-\infty\},
\mathord{\max}, \mathord+, -\infty, 0 \bigr)$.  Given two
semirings
\[ (\mathbb S, \mathord+, \mathord\cdot, 0, 1) \qquad \text{and}
  \qquad (\mathbb T, \mathord\oplus, \mathord\odot, \bot, \top)
  \enspace, \] a \emph{semiring homomorphism} is a
mapping~$h \in \mathbb T^{\mathbb S}$ such that~$h(0) = \bot$,
$h(1) = \top$, and $h(s_1 + s_2) = h(s_1) \oplus h(s_2)$ as well
as~$h(s_1 \cdot s_2) = h(s_1) \odot h(s_2)$ for
all~$s_1, s_2 \in \mathbb S$.  When there is no risk of confusion, we
refer to a semiring~$(\mathbb{S}, \mathord+, \mathord\cdot, 0, 1)$
simply by its carrier set~$\mathbb{S}$.  A semiring~$\mathbb{S}$ is a
\emph{ring} if there exists~$-1 \in \mathbb S$ such that~$-1 + 1 = 0$.
Let $\Sigma$~be a ranked alphabet.  Any
mapping~$A \in \mathbb S^{\Tsigma}$ is called a \emph{weighted tree
  language} over~$\mathbb{S}$, and its support
is~$\supp(A) = \{t \in \Tsigma \mid A_t \neq 0\}$.

Let $\Sigma$~and~$\Delta$ be ranked alphabets
and~$h' \in T_\Delta(X)^\Sigma$ a map such
that~$h'_\sigma \in T_\Delta(X_k) $ for all~$k \in \N$
and~$\sigma \in \Sigma_k$.  We extend~$h'$
to~$h \in T_\Delta^{\Tsigma}$ by
(i)~$h(\alpha) = h'_\alpha \in T_\Delta(X_0) = T_\Delta$ for
all~$\alpha \in \Sigma_0$ and
(ii)~$h\bigl(\sigma(\seq t1k) \bigr) = h'_\sigma \bigl[h(t_1), \dotsc,
h(t_k) \bigr]$ for all~$k \in \N$, $\sigma \in \Sigma_k$,
and~$\seq t1k \in \Tsigma$.  The mapping~$h$ is called the \emph{tree
  homomorphism induced by~$h'$}, and we identify~$h'$ and its induced
tree homomorphism~$h$.  It is \emph{nonerasing}
if~$h'_\sigma \notin X$ for all~$k\in\N$ and~$\sigma \in \Sigma_k$,
and it is \emph{nondeleting} if~$\var(h'_\sigma) = X_k$ for
all~$k\in\N$ and~$\sigma \in \Sigma_k$.  Let
$h \in T_\Delta^{\Tsigma}$ be a nonerasing and nondeleting
homomorphism.  Then $h$~is \emph{input finitary}; i.e., the
set~$h^{-1}(u)$ is finite for every~$u \in T_\Delta$
because~$\abs t \leq \abs u$ for each~$t \in h^{-1}(u)$.
Additionally, let~$A \in \mathbb{S}^{\Tsigma}$ be a weighted tree
language.  We define the weighted tree
language~$h(A) \in \mathbb{S}^{T_\Delta}$ for every~$u \in T_\Delta$
by $h(A)_u = \sum_{t \in h^{-1}(u)} A_t$.

\section{Weighted Tree Grammars with Constraints}
\label{sec:model}
Let us start with the formal definition of our weighted tree grammars.
They are a weighted variant of the tree automata with equality and
inequality constraints originally introduced
in~\cite{bogtis92,comjaq94}.  Compared to~\cite{bogtis92,comjaq94} our
model is slightly more expressive as we allow arbitrary constraints,
whereas constraints were restricted to subtrees occurring in the
productions in~\cite{bogtis92,comjaq94}.  This more restricted version
will be called classic in the following.  An overview of further
developments for these automata can be found in~\cite{tis11}.  We
essentially use the version recently utilized to solve the HOM
problem~\cite[Definition~4.1]{godoy2013hom}.  For the rest of this
section, let~$(\mathbb S, \mathord+, \mathord\cdot, 0, 1)$ be a
commutative semiring.

\begin{definition}[\protect{see~\cite[Definition~4.1]{godoy2013hom}}]
  A \emph{weighted tree grammar with constraints} (WTGc) is
  a tuple~$G  = (Q, \Sigma, F, P, \mathord{\wt})$ such that
  \begin{itemize}
  \item $Q$~is a finite set of nonterminals and~$F \in \mathbb S^Q$
    assigns final weights,
  \item $\Sigma$~is a ranked alphabet of input symbols,
  \item $P$~is a finite set of productions of the form~$(\ell, q, E,
    I)$, where~$\ell \in T_\Sigma(Q) \setminus Q$, $q \in Q$, and~$E,
    I \subseteq \nat^* \times \nat^*$ are finite sets, and
  \item $\mathord{\wt} \in \mathbb S^P$~assigns a weight to
    each production. \qed
  \end{itemize}
\end{definition}

In the following, let~$G = (Q, \Sigma, F, P, \mathord{\wt})$ be a
WTGc.  The components of a production~$p = (\ell, q, E, I) \in P$ are
the left-hand side~$\ell$, the target nonterminal~$q$, the set~$E$
of equality constraints, and the set~$I$ of inequality constraints.
Correspondingly, the production~$p$ is also
written~$\ell \stackrel{E,I}\longrightarrow q$ or
even~$\ell \stackrel{E,I}\longrightarrow_{\wt_p} q$ if we want to
indicate its weight.  Additionally, we simply list an equality
constraint~$(v, v') \in E$ as~$v = v'$ and an inequality
constraint~$(v, v') \in I$ as~$v \neq v'$.  A
production~$\ell \stackrel{E,I}\longrightarrow q \in P$ is
\emph{normalized} if~$\ell = \sigma(\seq q1k)$ for some~$k \in \nat$,
$\sigma \in \Sigma_k$, and~$\seq q1k \in Q$.  It is \emph{positive}
if~$I = \emptyset$; i.e., it has no inequality constraints, and it is
\emph{unconstrained} if~$E = \emptyset = I$; i.e., the production has
no constraints at all.  Instead of~$\ell \stackrel{\emptyset,
  \emptyset}\longrightarrow q$ we also write just~$\ell \to q$.  The
production is \emph{classic} if~$\{v, v'\} \subseteq \pos_Q(\ell)$ for all
constraints~$(v, v') \in E \cup I$.  In other words, in a classic
production the constraints can only refer to nonterminal-labeled
subtrees of the left-hand side.  The WTGc~$G$ is a \emph{weighted tree
  automaton with constraints}~(WTAc) if all productions~$p \in P$ are
normalized, and it is a \emph{weighted tree
  grammar}~(WTG)~\cite{gecste15} if all productions~$p \in P$ are
unconstrained.  If $G$~is both a WTAc as well as a WTG, then it is a
\emph{weighted tree automaton}~(WTA)~\cite{gecste15}.  All these
devices have \emph{Boolean final weights} if~$F \in \{0,1\}^Q$, they
are \emph{positive} if every~$p \in P$ is positive, and they are
\emph{classic} if every production~$p \in P$ is classic.  Finally, if
we utilize the Boolean semiring~$\mathbb B$, then we reobtain the
unweighted versions and omit the `W' in the abbreviations and the
mapping~`$\wt$' from the tuple.

The semantics for our WTGc~$G$ is a slightly non-standard
\emph{derivation semantics} when compared to~\cite[Definitions~4.3 \&
4.4]{godoy2013hom}.  Let~$(v,v') \in \N^* \times \N^*$
and~$t \in \Tsigma$.  If~$v,v' \in \pos(t)$ and~$t|_v = t|_{v'}$, we
say that~$t$ satisfies~$(v,v')$, otherwise~$t$ dissatisfies~$(v,v')$.
Let now~$C \subseteq \N^* \times \N^*$ be a finite set of constraints.
We write~$t\models C$ if~$t$ satisfies all~$(v,v') \in C$,
and~$t \nmodels C$ if~$t$ dissatisfies all~$(v,v') \in C$. Universally
dissatisfying~$C$ is generally stronger than simply not
satisfying~$C$.

\begin{definition}
  \label{df:derivsem}
  A \emph{sentential form (for~$G$)} is simply a tree
  of~$\xi \in T_\Sigma(Q)$.  Given an input tree~$t \in T_\Sigma$,
  sentential forms~$\xi, \zeta \in T_\Sigma(Q)$, a
  production~$p = \ell \stackrel{E,I}\longrightarrow q \in P$, and a
  position~$w \in \pos(\xi)$, we write~$\xi \Rightarrow_{G,t}^{p,w}
  \zeta$ if~$\xi|_w = \ell$, $\zeta = \xi[q]_w$, and the constraints
  $E$~and~$I$ are fulfilled on~$t|_w$; i.e., $t|_w \models
  E$~and~$t|_w \nmodels I$.  A sequence
  \[ d = (p_1,w_1) \dotsm (p_n, w_n) \in (P \times \nat^*)^* \]
  is a \emph{derivation of~$G$ for~$t$}
  if there exist~$\seq \xi1n \in T_\Sigma(Q)$ such that
  \[ t \Rightarrow_{G, t}^{p_1,w_1} \xi_1 \Rightarrow_{G, t}^{p_2,w_2}
    \dotsb \Rightarrow_{G, t}^{p_n,w_n} \xi_n \enspace. \]  It is
    \emph{left-most} 
  if additionally~$w_1 \prec w_2 \prec \dotsb \prec w_n$, where
  $\preceq$~is the lexicographic order on~$\nat^*$ in which prefixes
  are larger, so~$\varepsilon$ is the largest element.  \qed
\end{definition}

Note that the sentential forms~$\seq \xi1n$ are uniquely determined if
they exist, and for any derivation~$d$ for~$t$ there exists a unique
permutation of~$d$ that is a left-most derivation for~$t$.  The
derivation~$d$ is \emph{complete} if~$\xi_n  \in Q$, and in that case
it is also called a derivation to~$\xi_n$.  The set of all complete
left-most derivations for~$t$ to~$q \in Q$ is denoted by~$D^q_G(t)$.
The WTGc~$G$ is \emph{unambiguous} if~$\sum_{q \in \supp(F)}
\abs{D_G^q(t)} \leq 1$ for every~$t \in \Tsigma$.

Let~$p = \ell \stackrel{E,I}\longrightarrow q \in P$ be a production.
Since there exist unique~$k = \abs{\pos_Q(\ell)}$, $c \in \widehat
C_\Sigma(X_k)$, and~$\seq q1k \in Q$ such that~$\ell = c[\seq q1k]$,
we also simply write
\[ c[\seq q1k] \stackrel{E, I}\longrightarrow q \]
instead of~$p$.  Using this notation, we can present a recursion for
the set~$D_G^q(t)$ of complete derivations for~$t \in T_\Sigma$ to~$q 
\in Q$.
\begin{align*}
  D_G^q(t) = \Bigl\{ \word d1k (p,\varepsilon) \;
  &\Big|\; k \in \nat,\, p = c[\seq q1k] \stackrel{E,I}\longrightarrow
    q \in P,\, t \models E,\, t \nmodels I \\
  &\phantom{\Big|\;} \seq t1k \in T_\Sigma,\, t = c[\seq t1k],\,
    \forall i \in [k] \colon d_i \in D_G^{q_i}(t_i) \Bigr\}
\end{align*}
Specifically, let~$d = (p_1, w_1) \dotsm (p_n, w_n)$ be a complete
derivation for some tree~$t \in \Tsigma$.  For a given
position~$w \in \{\seq w1n\}$, we let~$k \in \nat$ and
$1 \leq i_1 < \dotsb < i_k \leq n$ be the indices such that
$\bigl\{\seq i1k \bigr\} = \bigl\{ i \in [n] \mid w_i = ww'_i
\bigr\}$; i.e., the indices of the derivation steps applied to
positions below~$w$ with~$w'_i$ being the suffix of~$w_i$ following
the prefix~$w$ for all~$i \in \{\seq i1k\}$.  The \emph{derivation
  for~$t|_w$ incorporated in~$d$} is the
derivation~$(p_{i_1}, w'_{i_1}), \dotsc, (p_{i_k}, w'_{i_k})$.
Conversely, for every~$w \in \N^*$ we abbreviate the
derivation~$(p_1, ww_1) \dotsm (p_n, ww_n)$ by simply~$wd$.

\begin{definition}
  The \emph{weight} of a derivation~$d = (p_1,w_1) \dotsm (p_n, w_n)$
  is defined to be
  \[ \wt_G(d) = \prod_{i = 1}^n \wt(p_i) \enspace. \]  The
  weighted tree language generated by~$G$, written simply~$G \in
  \mathbb S^{\Tsigma}$, is defined for every~$t \in \Tsigma$ by
  \[ G_t = \sum_{q \in Q,\, d \in D^q_G(t)} F_q \cdot \wt_G(d)
    \enspace. \tag*{\qed} \]
\end{definition}

Two WTGc are \emph{equivalent} if they generate the same weighted tree
language.  Finally, a weighted tree language is
\begin{itemize}
\item \emph{regular} if it is generated by some WTG, 
\item \emph{positive constraint-regular} if it is generated by some
  positive WTGc,
\item \emph{classic constraint-regular} if it is generated by some
  classic WTGc, and
\item \emph{constraint-regular} if it is generated by some WTGc.
\end{itemize}
Since the weights of productions are
multiplied,  we can assume without loss of generality that~$\wt_p \neq
0$ for all~$p \in P$.

\begin{example}
  \label{ex:1}
  Consider the WTGc~$G = (Q, \Sigma, F, P, \mathord{\wt})$
  over the arctic semiring~$\mathbb A$ with nonterminals~$Q = \{q, q'\}$,
  $\Sigma = \{\alpha^{(0)}, \gamma^{(1)}, \sigma^{(2)}\}$,
  $F_q = -\infty$, $F_{q'} = 0$, and $P$~and~`$\wt$' given by the
  productions $p_1 = \alpha \to_0 q$, $p_2 = \gamma(q) \to_1 q$, and
  $p_3 = \sigma\bigl(\gamma(q), q\bigr)
  \stackrel{11=2}\longrightarrow_1 q'$.  Clearly, $G$~is positive and
  classic, but not a WTAc.  The
  tree~$t = \sigma\bigl(\gamma(\gamma(\alpha)), \gamma(\alpha) \bigr)$
  has the unique left-most derivation
  \[ d = (p_1, 111) \, (p_2, 11) \, (p_1, 21) \, (p_2, 2) \, (p_3,
    \varepsilon) \]
  to the nonterminal~$q'$, which is illustrated in
  Figure~\ref{fig:deriv1}.  Overall, we have
  \[ \supp(G) = \big\{ \sigma\bigl(\gamma^{i+1}(\alpha),
    \gamma^i(\alpha)\bigr) \mid i \in \N \bigr\} \]
  and $G_t = \abs{\pos_\gamma(t)}$ for every~$t \in \supp(G)$, where
  $\gamma^i(t)$~abbreviates~$\gamma( \dotsm \gamma(t)
  \dotsm)$ containing $i$"~times the unary symbol~$\gamma$ atop~$t$.
  \qed
\end{example}

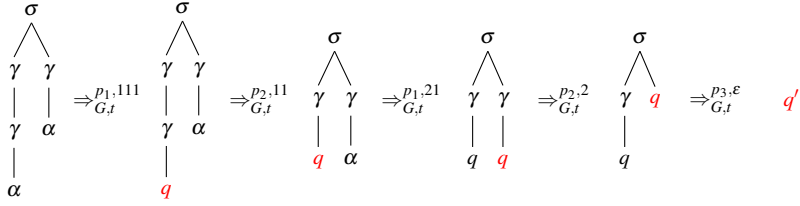
\begin{figure}
  \centering
  \begin{tikzpicture}
    \tikzset{level distance=0.8cm}
    \node at (0,0) {\Tree [.$\sigma$ [.$\gamma$ [.$\gamma$ $\alpha$ ] ]
      [.$\gamma$ $\alpha$ ] ]};
    \node at (1,0) {$\Rightarrow_{G,t}^{p_1,111}$};
    \node at (3,0) {$\Rightarrow_{G,t}^{p_2,11}$};
    \node at (5,0) {$\Rightarrow_{G,t}^{p_1,21}$};
    \node at (7,0) {$\Rightarrow_{G,t}^{p_2,2}$};
    \node at (9,0) {$\Rightarrow_{G,t}^{p_3,\varepsilon}$};
    \node at (2,0) {\Tree [.$\sigma$ [.$\gamma$ [.$\gamma$
      \textcolor{red}{$q$} ] ]
      [.$\gamma$ $\alpha$ ] ]};
    \node at (4,0) {\Tree [.$\sigma$ [.$\gamma$ \textcolor{red}{$q$} ]
      [.$\gamma$ $\alpha$ ] ]};
    \node at (6,0) {\Tree [.$\sigma$ [.$\gamma$ $q$ ] [.$\gamma$
      \textcolor{red}{$q$} ] ]};
    \node at (8,0) {\Tree [.$\sigma$ [.$\gamma$ $q$ ]
      \textcolor{red}{$q$} ]};
    \node at (10,0) {\textcolor{red}{$q'$}};
  \end{tikzpicture}
  \caption{Illustration of the derivation mentioned in
    Example~\protect{\ref{ex:1}}.}
  \label{fig:deriv1}
\end{figure}

Next, we introduce another semantics, called initial algebra
semantics, which is based on the presented recursive presentation of
derivations and often more convenient in proofs.

\begin{definition}
  For every nonterminal~$q \in Q$ we recursively define the
  map~$\mathord{\wt_G^q} \in \mathbb S^{\Tsigma}$ such that for
  every~$t \in T_\Sigma$ by
  \begin{equation}
    \label{eq:s1}
    \wt_G^q(t) = \sum_{\substack{p = c[\seq q1k] \stackrel{E,
          I}\longrightarrow q \in P \\ \seq t1k \in T_\Sigma \\ t =
        c[\seq t1k] \\ t \models E,\, t \nmodels I}} \wt_p \cdot
    \prod_{i = 1}^k \wt_G^{q_i}(t_i) \enspace.
  \end{equation}
  \mbox{} \qed
\end{definition}

It is a routine matter to verify that~$\wt_G^q(t) = \sum_{d \in
  D^q_G(t)} \wt_G(d)$ for every~$q \in Q$ and~$t \in \Tsigma$.  This
utilizes the presented recursive decomposition of complete derivations
as well as distributivity of the semiring~$\mathbb S$.
  
As for WTG~and~WTA~\cite{fulvog09}, also every (positive) WTGc can be
turned into an equivalent (positive) WTAc at the expense of additional
nonterminals by decomposing the left-hand sides.

\begin{lemma}[cf.~\protect{\cite[Lemma~4.8]{godoy2013hom}}]
  \label{lm:norm}
  WTGc and WTAc are equally expressive.  This also applies to
  positive WTGc.
\end{lemma}

\begin{proof}
  Let $G = (Q, \Sigma, F, P, \mathord{\wt})$ be a WTGc with a
  non-normalized production
  \[ p = \sigma(\seq \ell1k) \stackrel{E,I}\longrightarrow q \in P
    \enspace, \] let~$U\supseteq Q$ and let~$\varphi \in
  U^{T_\Sigma(Q)}$ be an injective map such that~$\varphi_q = q$
  for all~$q \in Q$.  We define the WTGc~$G' = (Q', \Sigma, F', P',
  \mathord{\wt'})$ such that~$Q' = Q \cup \{\varphi_{\ell_1}, \dotsc,
  \varphi_{\ell_k}\}$, $F'_q = F_q$~for all~$q \in Q$
  and~$F'_{q'} = 0$ for all~$q' \in Q' \setminus Q$, and
  \[ P' = \bigl(P \setminus \{p\} \bigr) \cup
    \bigl\{\sigma(\varphi_{\ell_1}, \dotsc, \varphi_{\ell_k})
    \stackrel{E,I}\longrightarrow q \bigr\} \cup \bigl\{\ell_i \to
    \varphi_{\ell_i} \mid i \in [k], \ell_i \notin Q \bigr\}
    \enspace, \]
  and for every~$p' \in P'$
  \[ \wt'_{p'} =
    \begin{cases}
      \wt_{p'}
      & \text{if } p' \in P \setminus \{p\} \\
      \wt_p
      & \text{if } p' = \sigma(\varphi_{\ell_1}, \dotsc,
      \varphi_{\ell_k}) \stackrel{E,I}\longrightarrow q \\
      1
      & \text{otherwise.}
    \end{cases} \]
  To prove that~$G'$ is equivalent to~$G$ we observe that for every
  left-most derivation
  \[ d = (p_1, w_1) \dotsm (p_n, w_n) \] of~$G$,
  there exists a corresponding derivation~$d'$ of~$G'$, which is
  obtained by replacing each derivation step~$(p_a, w_a)$
  with~$p_a = p$ by the sequence
  \[ (\ell_i \to \varphi_{\ell_i},\, w_ai)_{i \in [k], \ell_i \notin
      Q} \bigl(\sigma(\varphi_{\ell_1}, \dotsc, \varphi_{\ell_k})
    \stackrel{E,I}\longrightarrow q, w_a \bigr) \] of 
  derivation steps 
  of~$G'$ (yielding also a unique corresponding left-most
  derivation).  This replacement preserves the weight of the
  derivation.  Vice versa any left-most derivation of~$G'$ that
  utilizes the production~$\sigma(\varphi_{\ell_1}, \dotsc,
  \varphi_{\ell_k}) \stackrel{E,I}\longrightarrow q \in P'$ 
  at~$w$ needs to previously utilize the productions~$\ell_i \to
  \varphi_{\ell_i} \in P'$ at~$wi$ for all~$i \in [k]$ with~$\ell_i
  \notin Q$ since these are the only productions that
  generate the nonterminal~$\varphi_{\ell_i}$.  Thus, we established a
  weight-preserving bijection between the left-most derivations of
  $G$~and~$G'$, so it is obvious that~$G' = G$.  Repeated application
  of the normalization eventually (after finitely many steps) yields
  an equivalent WTAc.  Finally, we note that the constructed WTAc is
  positive if the original WTGc is positive. \qed
\end{proof}

As we will see in the next example, the construction used in the proof
of Lemma~\ref{lm:norm} does not preserve the classic property.

\begin{example}
  Consider the classic and positive WTGc~$G$ of Example~\ref{ex:1} and
  its non-normalized production~$p = \sigma\big(\gamma(q), q\big)
  \stackrel{11=2}\longrightarrow_1 q'$.  Applying the construction in
  the proof of Lemma~\ref{lm:norm} we replace~$p$ by the productions
  $\sigma(q'', q) \stackrel{11=2}\longrightarrow_1 q$, which is not
  classic, and~$\gamma(q) \to_0 q''$, where $q''$~is some new
  nonterminal.  The WTGc obtained this way is already a positive
  WTAc. \qed
\end{example}

Another routine normalization turns the final weights into Boolean
final weights following the approach of~\cite[Lemma~6.1.1]{bor05a}.
This is achieved by adding special copies of all nonterminals that
terminate the derivation and pre-apply the final weight.

\begin{lemma}
  \label{lm:rootweights}
  WTGc and WTGc with Boolean final weights are equally expressive.
  This also applies to positive WTGc, classic WTGc, and classic
  positive WTGc as well as the same WTAc.
\end{lemma}

\begin{proof}
  Let $G = (Q, \Sigma, F, P, \mathord{\wt})$ be a WTGc.  Let~$f \in
  C^Q$ be bijective with~$C \cap Q = \emptyset$.  We construct
  the WTGc~$G' = (Q \cup C, \Sigma, F', P \cup P', \mathord{\wt} \cup
  \mathord{\wt'})$ such that~$p' = \ell \stackrel{E, I}\longrightarrow
  f_q$ belongs to~$P'$ and~$\wt'_{p'} = \wt_p \cdot F_q$ for every~$p =
  \ell \stackrel{E, I}\longrightarrow q \in P$.  No other productions
  belong to~$P'$.  Finally, $F'_q = 0$ for all~$q \in Q$ and $F_c = 1$
  for all~$c \in C$.  The proof of equivalence is straightforward
  showing for every~$t \in T_\Sigma$ and $q \in Q$ that
  \[ \wt_{G'}^q(t) = \wt_G^q(t) \qquad \text{and} \qquad
    \wt_{G'}^{f(q)}(t) = \wt_G^q(t) \cdot F_q \enspace. \]
  The construction trivially preserves the properties normalized,
  positive, and classic.  \qed
\end{proof}

Let~$d \in D_G^q(t)$ be a derivation for some~$q \in Q$ and~$t \in
\Tsigma$.  Since we often argue with the help of such derivations~$d$,
it is a nuisance that we might have~$\wt_G(d) = 0$.  This anomaly can
occur even if~$\wt_p \neq 0$ for all~$p \in P$ due to the presence of
zero-divisors, which are elements~$s, s' \in \mathbb S \setminus
\{0\}$ such that~$s \cdot s' = 0$.  However, we can fortunately avoid
such anomalies altogether utilizing a construction of~\cite{kir11},
which has been lifted to tree automata in~\cite{droheu15}.

\begin{lemma}
  \label{lm:zero}
  For every WTGc~$G$ there exists
  a WTGc~$G' = (Q', \Sigma, F', P', \mathord{\wt'})$ that is
  equivalent and $\wt'_{G'}(d') \neq 0$ for all~$q' \in Q'$, $t' \in
  \Tsigma$, and~$d' \in D_{G'}^{q'}(t')$.  This also applies to
  positive WTGc, classic WTGc, and classic positive WTGc as well as
  the same WTAc.  The construction also preserves Boolean final
  weights.
\end{lemma}      

\begin{proof}
  Let $G = (Q, \Sigma, F, P, \mathord{\wt})$. 
  Obviously, $(\mathbb S, \mathord\cdot, 1, 0)$~is a commutative
  monoid with zero.  Let~$(\seq s1n)$ be an enumeration of the finite
  set~$\wt(P) \setminus \{1\} \subseteq \mathbb S$.  We consider the
  monoid homomorphism~$h \colon \nat^n \to \mathbb S$, which is given
  by
  \[ h(\seq m1n) = \prod_{i = 1}^n s_i^{m_i} \]
  for every~$\seq m1n 
  \in \nat$.  According to \textsc{Dickson}'s lemma~\cite{dic13} the
  set~$\min h^{-1}(0)$ is finite, where the partial order is the
  standard pointwise order on~$\nat^n$.  Hence there is~$u \in
  \nat$ such that~$\min h^{-1}(0) \subseteq \{0, \dotsc, u\}^n = U$.
  We define the operation~$\mathord{\oplus} \colon U^2 \to U$ by~$(v
  \oplus v')_i = \min(v_i + v'_i, u)$ for every~$v, v' \in U$ and~$i
  \in [n]$.  Moreover, for every~$i \in [n]$ we let~$1_{s_i} \in U$
  be the vector such that~$(1_{s_i})_i = 1$ and~$(1_{s_i})_a = 0$
  for all~$a \in [n] \setminus \{i\}$.  Let~$V = U \setminus
  h^{-1}(0)$.  We construct the equivalent WTGc~$G'$ such that $Q' = Q
  \times V$, $F'_{\langle q, v\rangle} = F_q$ for all~$\langle q, v
  \rangle \in Q'$, and $P'$~and~$\wt'$ are given as 
  follows.  For every production
  \[ p = c[\seq q1k] \stackrel{E,I}\longrightarrow q \in P \] and
  all~$\seq v1k \in V$ such that
  $v = 1_{\wt_p} \oplus \bigoplus_{i = 1}^k v_i \in V$ the production
  \[ c \bigl[\langle q_1, v_1\rangle, \dotsc, \langle q_k, v_k\rangle
    \bigr] \stackrel{E,I}\longrightarrow
    \langle q, v \rangle \] belongs to~$P'$ and its weight
  is~$\wt'_{p'} = \wt_p$.  No further productions are in~$P'$.  The
  construction trivially preserves the properties positive, classic,
  and normalized.  For correctness,
  let~$q' = \langle q, v\rangle \in Q'$, $t' \in \Tsigma$,
  and~$d' \in D_{G'}^{q'}(t')$.  We suitably (for the purpose of
  zero-divisors) track the weight of the derivation in~$v$
  and~$h_v \neq 0$ by definition.
  Consequently,~$\wt'_{G'}(d') \neq 0$ as required.  We note that
  possibly~$\wt_{G'}(d') \neq h_v$.  \qed
\end{proof}	

For zero-sum free semirings~\cite{gol99,hebwei98} we obtain that the
support~$\supp(G)$ of an WTGc can be generated by a TGc.  A semiring
is \emph{zero-sum free} if $s = 0 = s'$ for every $s, s' \in \mathbb
S$ such that~$s + s' = 0$.  Clearly, rings are never zero-sum free,
but the mentioned semirings~$\mathbb B$, $\mathbb N$, $\mathbb T$,
and~$\mathbb A$ are all zero-sum free.

\begin{corollary}[of Lemmata~\protect{\ref{lm:rootweights}
    and~\ref{lm:zero}}]
  \label{cor:regular}
  If~$\mathbb S$ is zero-sum free, then~$\supp(G)$ is
  (positive, classic) constraint-regular for every (respectively,
  positive, classic) WTGc~$G$.  
\end{corollary}

\begin{proof}
  We apply Lemma~\ref{lm:rootweights} to obtain an
  equivalent WTGc with Boolean final weights and then
  Lemma~\ref{lm:zero} to obtain the
  WTGc~$G' = (Q', \Sigma, F', P', \mathord{\wt'})$ with Boolean final
  weights.  As mentioned we can assume that~$\wt'_{p'} \neq 0$ for
  all~$p' \in P'$.  Let~$q' \in \supp(F')$ and~$t' \in \Tsigma$
  with~$D_{G'}^{q'}(t')\neq\emptyset$.  Since~$\wt'_{G'}(d') \neq 0$
  for every derivation~$d' \in D_{G'}^{q'}(t')$ and~$s + s' \neq 0$
  for all~$s, s' \in \mathbb S \setminus \{0\}$ due to zero-sum
  freeness, we obtain~$t' \in \supp(G')$.  Thus, the existence of a
  complete derivation for~$t'$ to an accepting nonterminal (i.e., one
  with final weight~$1$) characterizes whether we have~$t' \in
  \supp(G')$.  Consequently, the TGc~$\bigl(Q', \Sigma, \supp(F'), P'
  \bigr)$ generates the tree language~$\supp(G')$, which is thus
  constraint-regular.  The properties positive and classic are
  preserved in all the constructions.  \qed
\end{proof}

\section{Closure Properties}
\label{sec:closure}
Next we investigate several closure properties of the
constraint-regular weighted tree languages.  We start with the
(point-wise) sum, which is given by~$(A + A')_t = A_t + A'_t$ for
every~$t \in \Tsigma$ and~$A, A' \in \mathbb S^{\Tsigma}$.  Given WTGc
$G$~and~$G'$ generating $A$~and~$A'$ we can trivially use a disjoint
union construction to obtain a WTGc generating~$A + A'$.  We omit the
details.

\begin{proposition}
  \label{prop:sum}
  The (positive, classical) constraint-regular weighted tree languages
  (over a fixed ranked alphabet) are closed under sums. \qed
\end{proposition}

The corresponding (point-wise) product is the \textsc{Hadamard}
product, which is given by~$(A \cdot A')_t = A_t \cdot A'_t$
for every~$t \in \Tsigma$ and~$A, A' \in \mathbb S^{\Tsigma}$.  With
the help of a standard product construction we show that the
(positive) constraint-regular weighted tree languages are also closed
under \textsc{Hadamard} product.  As preparation we introduce a
special normal form.  A WTAc~$G = (Q, \Sigma, F, P, \mathord{\wt})$ is
\emph{constraint-determined} if $E = E'$~and~$I = I'$ for all
productions
\[ \sigma(\seq q1k) \stackrel{E, I}\longrightarrow q \in P \quad
  \text{and} \quad \sigma(\seq q1k) \stackrel{E', I'}\longrightarrow q
  \in P \enspace. \] 
In other words, two productions cannot differ only in the sets of
constraints.  It is straightforward to turn any (positive) WTAc into
an equivalent constraint-determined (positive) WTAc by introducing
additional nonterminals (e.g.~annotate the constraints to the
nonterminal on the right-hand side).

\begin{theorem}
  \label{thm:hadamard}
  The (positive) constraint-regular weighted tree languages (over a
  fixed ranked alphabet) are closed under \upshape{\textsc{Hadamard}}
  products.
\end{theorem}

\begin{proof}
  Let $A, A' \in \mathbb S^{\Tsigma}$ be constraint-regular.
  Without loss of generality (see Lemma~\ref{lm:norm}) we can assume
  constraint-determined WTAc
  \[ G = (Q, \Sigma, F, P, \mathord{\wt}) \qquad \text{and} \qquad G'
    = (Q', \Sigma, F', P', \mathord{\wt'}) \] 
  that generate $A$~and~$A'$, respectively.  We construct the
  direct product WTAc
  \[ G \times G' = (Q \times Q', \Sigma, F'', P'', \mathord{\wt''}) \]
  such that $F''_{\langle q, q'\rangle} = F_q \cdot F'_{q'}$ for
  every~$q \in Q$ and~$q' \in Q'$ and for every production $p =
  \sigma(\seq q1k) \stackrel{E, I}\longrightarrow q \in P$ and
  production~$p' = \sigma(\seq{q'}1k) \stackrel{E', I'}\longrightarrow
  q' \in P'$ the production
  \[ p'' = \sigma \bigl( \langle q_1, q'_1\rangle, \dotsc, \langle q_k,
    q'_k\rangle \bigr) \stackrel{E \cup E', I \cup I'}\longrightarrow
    \langle q, q'\rangle \] 
  belongs to~$P''$ and its weight is~$\wt''_{p''} = \wt_p \cdot
  \wt'_{p'}$.  No other productions belong to~$P''$.  It is
  straightforward to see that the property positive is preserved.  The
  correctness proof that~$G \times G' = A \cdot A'$ is a
  straightforward induction proving
  \[ \wt_{G \times G'}^{\langle q, q'\rangle}(t) = \wt_G^q(t)
  \cdot \wt_{G'}^{q'}(t) \] for all~$t \in T_\Sigma$ using the initial
  algebra semantics.  The WTAc $G$~and~$G'$ are required to be
  constraint-determined, so that we can uniquely identify the
  basic productions $p \in P$~and~$p' \in P'$ that construct a newly
  formed production~$p'' \in P''$. 
  
  We can obtain a constraint-determined WTAc at the expense of a
  polynomial increase in the number of productions (assuming that the
  ranked alphabet of input symbols is fixed).
  Let~$r = \max_{\sigma \in \Sigma} \rk(\sigma)$ be the maximal rank
  of an input symbol and $c = \abs P$~be the number of productions of
  the given WTAc~$G = (Q, \Sigma, F, P, \mathord{\wt})$.  First, we
  modify the target nonterminal~$q$ of each production~$\rho = (\ell, q, E,
  I) \in P$ to additionally include the identifier~$\rho$, which
  yields the production~$(\ell, \langle q, \rho\rangle, E, I)$.  This
  effectively yields the new nonterminal set~$Q \times P$, which has
  size~$\abs Q \cdot c$.  Then we create copies of the
  production~$(\sigma(\seq q1k), \langle q, \rho\rangle, E, I)$ by
  the set of productions
  \[ \Bigl\{ \bigl(\sigma(\langle q_1, \rho_1\rangle, \dotsc, \langle
    q_k, \rho_k \rangle), \langle q, \rho\rangle, E, I \bigr)
    \;\Big|\; \seq\rho1k \in P \Bigr\} \enspace. \]
  Clearly, this turns each production into at most~$c^r$ productions
  since~$k \leq r$, so the overall number of productions after all
  replacements is at most~$c^{r+1}$.  The product construction itself
  is then quadratic. \qed
\end{proof}

We note that the previous construction also works for classic WTAc.

\begin{example}
  \label{ex:union}
  Let $G = \bigl(\{q\}, \Sigma, F, P, \mathord{\wt} \bigr)$
  and $G' = \bigl(\{z\}, \Sigma, F', P', \mathord{\wt'} \bigr)$
  be WTAc over~$\mathbb A$ and~$\Sigma = \{\alpha^{(0)}, \gamma^{(1)},
  \sigma^{(2)}\}$, $F_q = F'_z = 0$, and the productions
  \begin{align*}
    \alpha
    &\to_0 q
    & \gamma(q)
    &\to_2 q
    & \sigma(q, q)
    &\stackrel{1=2}\longrightarrow_0 q \tag{$P$} \\*
    \alpha
    &\to_0 z
    & \gamma(z)
    &\stackrel{11 \neq 12}\longrightarrow_1 z
    & \sigma(z, z)
    &\to_1 z \:. \tag{$P'$}
  \end{align*}
  We observe that
  \begin{align*}
    \supp(G)
    &= \bigl\{t \in \Tsigma \mid \forall w \in \pos_\sigma(t) \colon
      t|_{w1} = t|_{w2} \bigr\} \\
    \supp(G')
    &= \bigl\{t \in \Tsigma \mid \forall w \in \pos_\gamma(t) \colon
      \text{ if } t(w1) = \sigma \text{ then } t|_{w11} \neq t|_{w12}
      \bigr\}
  \end{align*}
  and~$G_t = 2\abs{\pos_\gamma(t)}$ as well
  as~$G'_{t'} = \abs{\pos_\gamma(t')} + \abs{\pos_\sigma(t')}$ for
  every tree~$t \in \supp(G)$ and tree~$t' \in \supp(G')$.  We obtain
  the
  WTAc~$G \times G' = \bigl(\{\langle q,z\rangle\}, \Sigma, F'', P'',
  \mathord{\wt''}\bigr)$ with $F''_{\langle q, z\rangle} = 0$ and the
  following productions.
  \[ \alpha \to_0 \langle q,z\rangle \qquad \gamma\bigl(\langle
    q,z\rangle \bigr) \stackrel{11\neq 12}\longrightarrow_3 \langle q,
    z\rangle \qquad \sigma\bigl( \langle q,z\rangle, \langle q,z\rangle 
    \bigr) \stackrel{1=2}\longrightarrow_1 \langle q,z\rangle \]
  Hence we obtain the equality~$(G \times G')_t =
  3\abs{\pos_\gamma(t)} + \abs{\pos_\sigma(t)} = G_t \cdot G'_t$ for
  every tree~$t \in \supp(G) \cap \supp(G')$.  \qed 
\end{example}

Next, we use an extended version of the classical power set
construction to obtain an unambiguous WTAc that keeps track of the
reachable nonterminals, but preserves only the homomorphic image of
its weight.  The unweighted part of the construction
mimics a power-set construction and the handling of constraints roughly
follows~\cite[Definition~3.1]{godoy2013hom}.

\begin{theorem}
  \label{thm:unamb}
  Let $h \in \mathbb T^{\mathbb S}$ be a semiring homomorphism into a
  finite semiring~$\mathbb T$.  For every (classic) WTAc~$G
  = (Q, \Sigma, F, P, \mathord{\wt})$ over $\mathbb S$ there exists an
  unambiguous (classic) WTAc~$G' = (\mathbb T^Q, \Sigma, F',
  P', \mathord{\wt'})$ such that for every tree~$t \in \Tsigma$
  and~$\varphi \in \mathbb T^Q$ 
  \[ \wt_{G'}^\varphi(t) =
    \begin{cases}
      1
      & \emph{if } \varphi_q = h \bigl(\wt_G^q(t) \bigr) \emph{ for
        all } q \in Q \\
      0
      & \emph{otherwise.}
    \end{cases} \]
  Moreover, $G'_t = h(G_t)$ for every~$t \in \Tsigma$.
\end{theorem}

\begin{proof}
  For every~$\sigma \in \Sigma$, let
  \[ \mathcal C_\sigma = \bigl\{ E \mid \sigma(\seq q1k)
    \stackrel{E,I}\longrightarrow q \in P \bigr\} \cup \bigl\{ I \mid
    \sigma(\seq q1k) \stackrel{E,I}\longrightarrow q \in P \bigr\} \]
  be the constraints that occur in productions of~$G$ whose left-hand
  side contains~$\sigma$.  We
  let~$F'_\varphi = \sum_{q \in Q} h(F_q) \cdot \varphi_q$ for
  every~$\varphi \in \mathbb T^Q$.  For all~$k \in \N$,
  $\sigma \in \Sigma_k$,
  nonterminals~$\varphi^1, \dotsc, \varphi^k \in \mathbb T^Q$, and
  constraints~$\mathcal E \subseteq \mathcal C_\sigma$ we let
  $p' = \sigma(\varphi^1, \dotsc \varphi^k) \stackrel{\mathcal E,
    \mathcal I}\longrightarrow \varphi \in P'$, where
  $\mathcal I = \mathcal C_\sigma \setminus \mathcal E$ and for
  every~$q \in Q$
  \begin{equation}
    \label{eq:1}
    \varphi_q = \sum_{\substack{p = \sigma(\seq q1k) \stackrel{E,
          I}\longrightarrow q \in P \\ E \subseteq \mathcal E,\, I
        \subseteq \mathcal I}} h(\wt_p) \cdot \varphi^1_{q_1}
    \cdot \ldots \cdot \varphi^k_{q_k} \enspace.
  \end{equation}
  No additional productions belong to~$P'$.  Finally, we
  set~$\mathord{\wt'_{p'}} = 1$ for all~$p' \in P'$.
  In general, the WTAc~$G'$ is certainly not deterministic due to the
  choice of constraints, but $G'$~is unambiguous since the
  resulting $2^{\abs{\mathcal C_\sigma}}$~rules for each left-hand side have
  mutually exclusive constraint sets.  In fact, for each~$t\in
  \Tsigma$ there is exactly one left-most complete derivation of~$G'$
  for~$t$, and it derives to~$\varphi \in \mathbb T^Q$ such that
  $\varphi_q = h \bigl(\wt_G^q(t) \bigr)$ for every~$q \in Q$.  The
  weight of that derivation is~$1$.  These statements are proven
  inductively.   The final statement~$G'_t = h(G_t)$ for every~$t \in
  T_\Sigma$ is an easy consequence of the previous statements.   If
  $G$~is classic, then also the constructed WTAc~$G'$ is classic.
  \qed
\end{proof}

\begin{example}
  \label{ex:disambiguation}
  Recall the WTAc~$G$~and~$G'$ from Example~\ref{ex:union}. 
  Consider the WTAc generating their disjoint union, as well as 
  the semiring homomorphism~$h \in \mathbb B^{\mathbb A}$ given 
  by~$h_a = 1$ for all~$a \in \mathbb A \setminus \{-\infty\}$ 
  and~$h_{-\infty} = 0$.
  The sets $\mathcal C_\gamma$~and~$\mathcal C_\sigma$ of utilized
  constraints are~$\mathcal C_\gamma = \bigl\{(11, 12) \bigr\}$ and
  $\mathcal C_\sigma = \bigl\{(1,2) \bigr\}$, and we write~$\varphi
  \in \mathbb B^Q$ simply as 
  subsets of~$Q$.  We obtain the unambiguous WTAc~$G''$ with the
  following sensible (i.e., having satisfiable constraints)
  productions for all~$Q', Q'' \subseteq \{q, z\}$, which all have
  weight~$1$.
  \begin{align*}
    \alpha
    &\stackrel{\phantom{11\neq 12}}\longrightarrow \{q, z\}
    \\*
    \gamma(Q')
    &\stackrel{11 = 12}\longrightarrow Q' \cap \{q\}
    & \gamma(Q')
    &\stackrel{11 \neq 12}\longrightarrow Q' \\*
    \sigma(Q', Q'')
    &\stackrel{\phantom11=2\phantom1}\longrightarrow Q' \cap Q''   
    &\sigma(Q', Q'')
    &\stackrel{\phantom11 \neq 2\phantom1}\longrightarrow Q' \cap Q''
      \cap \{z\}
  \end{align*}
  Each~$t \in \Tsigma$ has exactly one left-most complete derivation
  in~$G''$; it derives to~$Q'$, where (i)~$q \in Q'$ iff~$t \in
  \supp(G)$ and (ii)~$z \in Q'$ iff $t \in \supp(G')$. It 
  is~$F''_\emptyset=0$ and~$F''_Q=1$ for all non-empty~$Q\subseteq
  \{q,z\}$.\qed
\end{example}

\begin{corollary}[of Theorem~\protect{\ref{thm:unamb}}]
  Let $\mathbb S$ be finite.  For every (classic) WTAc over~$\mathbb
  S$ there exists an equivalent unambiguous (classic) WTAc. \qed 
\end{corollary}

\begin{corollary}[of Theorem~\protect{\ref{thm:unamb}}]
  \label{cor:supp}
  Let $\mathbb S$ be zero-sum free.  For every (classic) WTAc~$G$
  over~$\mathbb S$ there exists an unambiguous (classic) TAc
  generating~$\supp(G)$.
\end{corollary}

\begin{proof}
  Utilizing Lemma~\ref{lm:rootweights} we can first construct an
  equivalent WTAc with Boolean final weights.  If~$\mathbb S$ is
  zero-sum free, then there exists a semiring homomorphism~$h \in
  \mathbb B^{\mathbb S}$ by~\cite{wan97}.  By Lemma~\ref{lm:zero} we
  can assume that each derivation of~$G$ has non-zero weight and sums
  of non-zero elements remain non-zero by zero-sum freeness.  Thus we
  can simply replace the factor~$h(\wt_p)$ by~$1$ in~\eqref{eq:1}.
  The such obtained TAc generates~$\supp(G)$. \qed
\end{proof}

\begin{corollary}[of Theorem~\protect{\ref{thm:unamb}}]
  Let $\mathbb S$ be zero-sum free.  For every (classic) WTAc~$G$
  over~$\mathbb S$ there exists an unambiguous (classic) TAc
  generating~$T_\Sigma \setminus \supp(G)$. 
\end{corollary}

\begin{proof}
  Let~$G' = (Z, \Sigma, Z_0, P')$ be the unambiguous TAc given by
  Corollary~\ref{cor:supp}. Since~$G'$ is also complete in the sense
  that every input tree has a derivation, the desired unambiguous
  TAc~$G''$ is simply~$G'' = (Z, \Sigma, Z\setminus Z_0, P')$. \qed
\end{proof}

%
Let~$A, A' \in \mathbb S^{\Tsigma}$.  It is often
useful~(see~\cite[Definition~4.11]{godoy2013hom}) to restrict~$A$ to
the support of~$A'$ but without changing the weights of those trees
inside the support.  Formally, we define~$A|_{\supp(A')} \in \mathbb
S^{\Tsigma}$ for every~$t \in \Tsigma$ by~$A|_{\supp(A')}(t) = A_t$
if~$t \in \supp(A')$ and $A|_{\supp(A')}(t) = 0$ otherwise.  Utilizing
unambiguous WTAc and the \textsc{Hadamard} product, we can show
that $A|_{\supp(A')}$~is constraint-regular if $A$~and~$A'$
are constraint-regular and the semiring~$\mathbb S$ is zero-sum free.

\begin{theorem}
  Let $\mathbb S$ be zero-sum free.  For all (classic) WTAc
  $G$~and~$G'$ there exists a (classic) WTAc~$H$ such that~$H =
  G|_{\supp(G')}$.
\end{theorem}

\begin{proof}
  By Corollary~\ref{cor:regular} the support~$\supp(G')$ is
  constraint-regular.  Hence we can obtain an unambiguous WTAc~$G''$
  for~$\supp(G')$ using Theorem~\ref{thm:unamb}.  Without loss of
  generality we assume that both $G$~and~$G''$ are
  constraint-determined; we note that the normalization preserves
  unambiguous WTAc.  Finally we construct~$G \times G''$, which by
  Theorem~\ref{thm:hadamard} generates exactly~$G|_{\supp(G')}$ as
  required. \qed 
\end{proof} 

In the following, we establish a special property for classic
WTGc.  To this end, we first need another notion.  Let~$G = (Q,
\Sigma, F, P, \mathord{\wt})$ be a WTGc.  A nonterminal~$\bot \in Q$
is a \emph{sink nonterminal (in~$G$)} if $F_\bot = 0$ and
\[ \bigl\{ \sigma(\bot, \dotsc, \bot) \to_1 \bot \mid \sigma \in
  \Sigma \bigr\} = \bigl\{ \ell \stackrel{E,I}\longrightarrow_s q \in
  P \mid q = \bot \bigr\} \enspace. \] In other words, for every sink
nonterminal~$\bot$ the
production~$\sigma(\bot, \dotsc, \bot) \to \bot$ belongs to~$P$ with
weight~$1$ for every symbol~$\sigma \in \Sigma$.  Additionally, no
other productions have the sink nonterminal~$\bot$ as target
nonterminal.  Given a set~$E \subseteq \nat^* \times \nat^*$ of
equality constraints, we let~$\mathord{\equiv}_E = (E \cup E^{-1})^*$
be the smallest equivalence relation containing~$E$
and~$[w]_{\mathord{\equiv}_E}$ be the equivalence class
of~$w \in \nat^*$.  Additionally, for every production~$c[\seq q1k]
\stackrel{E,I}\longrightarrow q \in P$ we let
\[ c(E) = \bigl\{(i,j) \in [k] \times [k] \mid (v, v') \in E,\, c(v)
  = x_i,\, c(v') = x_j \bigr\} \]
be a representation of the equality constraints on the indices~$[k]$. 

\begin{definition}
  \label{df:eqrestricted}
  A classic WTGc~$G = (Q, \Sigma, F, P, \mathord{\wt})$ is
  \emph{eq-restricted} if there exists a sink nonterminal~$\bot \in Q$
  such that for every production~$p = c[\seq q1k]
  \stackrel{E,I}\longrightarrow q \in P$ and index~$i \in [k]$ there
  exists a nonterminal~$q' \in Q$ such that
  \begin{enumerate}
  \item $\{q_j \mid j \in [i]_{\equiv_{c(E)}} \} \subseteq \{q',
    \bot\}$ and
  \item there exists exactly one index~$j \in [i]_{\equiv_{c(E)}}$,
    also called \emph{governing index for~$i$ in~$p$}, such
    that~$q_j = q'$.
  \end{enumerate}
  The mapping~$g_p \colon [k] \to [k]$ assigns to each index~$i \in
  [k]$ its governing index for~$i$ in~$p$. \qed
\end{definition}

In other words, in an eq-restricted classic WTGc one subtree is
generated normally by the WTGc and all the subtrees that are required
to be equal by means of the equality constraints are generated by the
sink nonterminal~$\bot$, which can generate any tree with weight~$1$.
In this manner, the restrictions on subtree and weight generation
induced by the WTGc are exhibited completely on a single subtree and
the ``copies'' are only provided by the equality constraint, but not
further restricted by the WTGc.  We will continue to use~$\bot$ for
the suitable sink nonterminal of an eq-restricted classic WTGc.

Finally, we show that the weighted tree languages generated by
eq-restricted positive classic WTGc are closed under relabelings.  A
\emph{relabeling} is a tree
homomorphism~$\pi \in T_\Delta(X)^\Sigma$ such that for
every~$k \in \N$ and~$\sigma \in \Sigma_k$ there
exists~$\delta \in \Delta_k$ with~$\pi_\sigma = \delta(\seq x1k)$.
In other words, a relabeling deterministically replaces symbols
respecting their rank.  We often specify a relabeling just as a
mapping~$\pi \in \Delta^\Sigma$ such that~$\pi_\sigma \in
\Delta_k$ for every~$k \in \N$ and~$\sigma \in \Sigma_k$.

\begin{theorem}
  \label{thm:rel}
  The weighted tree languages generated by eq-restricted positive
  classic WTGc are closed under relabelings.
\end{theorem}

\begin{proof}
  Let WTGc~$G = (Q, \Sigma, F, P, \mathord{\wt})$ be an eq-restricted
  positive classic WTGc with sink nonterminal~$\bot$.  Without loss of
  generality, suppose that~$\Sigma \cap X = \emptyset$.  Moreover,
  let~$\pi \in \Delta^\Sigma$ be a relabeling.  We first
  extend~$\pi$ to a mapping~$\pi' \in (\Delta \cup X)^{\Sigma \cup
    X}$, in which we treat the elements of~$X$ as nullary symbols, 
  for every~$\sigma \in \Sigma$ and~$x \in X$ by $\pi'_\sigma =
  \pi_\sigma$ and $\pi'_x = x$.  Let $G' = (Q, \Delta, F, P',
  \mathord{\wt}')$ be the eq-restricted positive classic WTGc such
  that
  \begin{align*}
    P'
    &= \Bigl\{ \pi'(c)[\seq q1k] \stackrel{E,
      \emptyset}\longrightarrow q \mid c[\seq q1k] \stackrel{E, 
      \emptyset}\longrightarrow q \in P,\, q \neq \bot \Bigr\} \\
    &\phantom{{}={}} {} \cup \Bigl\{ \delta(\bot, \dotsc, \bot)
      \to \bot \mid \delta \in \Delta\Bigr\}
  \end{align*}
  and for every production~$p' = c'[\seq q1k] \stackrel{E,
    \emptyset}\longrightarrow q \in P'$ with~$q \neq \bot$ we let
  \begin{equation}
    \label{eq:w1}
    \wt'_{p'} = \sum_{\substack{p = c[\seq q1k]
        \stackrel{E,\emptyset}\longrightarrow q \in P \\ c \in
        (\pi')^{-1}(c')}} \wt_p \enspace.
  \end{equation}
  Finally, $\wt' \bigl(\delta(\bot, \dotsc, \bot) \to
  \bot \bigr) = 1$ for all~$\delta \in \Delta$. 
  For correctness we prove the following equality for every~$u \in
  T_\Delta$ and~$q \in Q$ by induction on~$u$
  \begin{equation}
    \label{eq:p2}
    \wt_{G'}^q(u) =
    \begin{cases}
      \sum_{t \in \pi^{-1}(u)} \wt_G^q(t)
      & \text{if } q \neq \bot \\
      1
      & \text{otherwise.}
    \end{cases}
  \end{equation}
  The second case is immediate since there is a single derivation,
  namely the one utilizing only nonterminal~$\bot$, for~$u$ to~$\bot$
  and its weight is~$1$.  In the remaining case we have~$q \neq
  \bot$.  Then
  \begin{align*}
    &\phantom{{}={}} \wt_{G'}^q(u) \\
    &\stackrel{\eqref{eq:s1}}= \sum_{\substack{p' = c'[\seq q1k]
      \stackrel{E, \emptyset}\longrightarrow q \in P' \\ \seq u1k \in
    T_\Delta \\ u = c'[\seq u1k] \\ u \models E}} \wt'_{p'} \cdot
    \prod_{i = 1}^k \wt_{G'}^{q_i}(u_i) \\
    &\stackrel{\text{IH}}= \sum_{\substack{p' = c'[\seq q1k]
      \stackrel{E, \emptyset}\longrightarrow q \in P' \\ \seq u1k \in
    T_\Delta \\ u = c'[\seq u1k] \\ u \models E}} \wt'_{p'} \cdot
    \prod_{\substack{i \in [k] \\ q_i \neq \bot}} \Bigl( \sum_{t_i \in
    \pi^{-1}(u_i)} \wt_G^{q_i}(t_i) \Bigr) \cdot \prod_{\substack{i
    \in [k] \\ q_i = \bot}} 1 \enspace. \\
    \intertext{Recall that~$g_p \colon [k] \to [k]$ assigns to each
    index its governing index.  For better readability, we write
    just~$g'$.  Note that due to the special form of substitution we
    automatically fulfill~$u \models E$ and can thus drop it.}
    &\stackrel{\eqref{eq:w1}}= \sum_{\substack{p' = c'[\seq q1k]
      \stackrel{E, \emptyset}\longrightarrow q \in P' \\ \forall i \in
    \ran(g') \colon u_i \in T_\Delta,\, t_i \in \pi^{-1}(u_i) \\ u =
    c'[u_{g'(1)}, \dotsc, u_{g'(k)}]}} \Bigl(\sum_{\substack{p =
    c[\seq q1k] \stackrel{E, \emptyset}\longrightarrow q \in P \\ c
    \in (\pi')^{-1}(c')}} \wt_p \Bigr) \cdot \prod_{i \in \ran(g')}
    \wt_G^{q_i}(t_i) \\
    \intertext{We note that~$g_{p'} = g_p$ for all used
    productions~$p$, so we just write~$g$.  Additionally, for
    every~$q_i$ with~$i \in [k] \setminus \ran(g)$ we have~$q_i =
    \bot$ and thus~$\wt_G^{q_i}(t_{g(i)}) = 1$ because there is
    exactly one such derivation with weight~$1$.}
    &= \sum_{\substack{p = c[\seq q1k] \stackrel{E,
      \emptyset}\longrightarrow q \in P \\ \forall i \in \ran(g)
    \colon t_i \in T_\Sigma \\ u = \pi(c[\seq t{g(1)}{g(k)}])}} \wt_p
    \cdot \prod_{i = 1}^k \wt_G^{q_i}(t_{g(i)}) \\
    &= \sum_{t \in \pi^{-1}(u)} \Biggl( \sum_{\substack{p = c[\seq
      q1k] \stackrel{E, \emptyset}\longrightarrow q \in P \\ \seq t1k
    \in T_\Sigma \\ t = c[\seq t1k] \\ t \models E}} \wt_p \cdot
    \prod_{i = 1}^k \wt_G^{q_i}(t_i) \Biggr)
    \stackrel{\eqref{eq:s1}}= \sum_{t \in \pi^{-1}(u)} \wt_G^q(t)
  \end{align*}
  We complete the proof for every~$u \in T_\Delta$ as follows.
  \begin{align*}
    G'_u
    &= \sum_{q \in Q} F_q \cdot \wt_{G'}^q(u) \stackrel{\eqref{eq:p2}}=
      \sum_{q \in Q \setminus \{\bot\}} F_q \cdot \Bigl(\sum_{t \in
      \pi^{-1}(u)} \wt_G^q(t) \Bigr) \\
    &= \sum_{t \in \pi^{-1}(u)} \Bigl(\sum_{q \in Q} F_q \cdot
      \wt_G^q(t) \Bigr) = \sum_{t \in \pi^{-1}(u)} G_t \tag*{\qed}
  \end{align*} 
\end{proof}

\section{Towards the HOM Problem}
\label{sec:decid}
The strategy of~\cite{godoy2013hom} for deciding the HOM problem first
represents the homomorphic image~$L' = h(L)$ of the regular tree
language~$L$ with the help of an WTGc~$G'$.  For deciding whether
$L'$~is regular, a tree automaton~$G''$ simulating the behavior
of~$G'$ up to a certain bounded height is constructed.  If the
automata $G'$~and~$G''$ are equivalent, i.e.,~$G'' = G'$, then $L'$~is
regular.  In the remaining case, pumping arguments are used to prove
that it is impossible to find any TA for~$L'$.  Overall, this reduces
the HOM problem to an equivalence problem.

Towards solving the HOM problem in the weighted case we now proceed
similarly.  First, we show that WTGc can encode each (well-defined)
homomorphic image of a regular weighted tree language.  This ability
motivated their definition in the unweighted
case~\cite[Proposition~4.6]{godoy2013hom}, and it also applies in the
weighted case with minor restrictions that just enforce that all
obtained sums are finite.

\begin{theorem}
  \label{thm:hom}
  Let $G = (Q, \Sigma, F, P, \mathord{\wt})$ be a WTA and $h \in
  T_\Delta^{\Tsigma}$~be a nondeleting and nonerasing tree homomorphism. 
  There exists an eq-restricted positive classic WTGc~$G'$ with $G' =
  h(G)$.
\end{theorem}

\begin{proof}
  We construct a WTGc~$G'$ for~$h(G)$ in two stages.  First, 
  let
  \[ G'' = \bigl(Q \cup \{\bot\}, \Delta \cup \Delta \times P, F'', P'',
    \mathord{\wt''} \bigr) \] such that for every~$p = \sigma(\seq
  q1k) \to q \in P$ and $h_\sigma = u = \delta(\seq u1n)$,
  \begin{align*}
    p''
    &= \Bigl( \langle \delta,p\rangle(\seq u1n) \llbracket \seq
      q1k\rrbracket \stackrel{E, \emptyset}\longrightarrow q \Bigr)
      \in P''
  \end{align*}
  with $E = \bigcup_{i \in [k]} \pos_{x_i}(u)^2$, in which the
  substitution~$\langle \delta,p\rangle(\seq u1n) \llbracket \seq q1k
  \rrbracket$ replaces for every~$i \in [k]$ only the left-most
  occurrence of~$x_i$ in~$\langle \delta,p\rangle(\seq u1n)$ by~$q_i$
  and all other occurrences by~$\bot$.
  Moreover~$\wt''_{p''} = \wt_p$.  Additionally, we
  let
  \[ p''_\delta = \delta(\bot, \dotsc, \bot) \to \bot \in P'' \] with
  weight~$\wt''_{p''_\delta} = 1$ for every~$k \in \nat$ and
  $\delta \in \Delta_k \cup \Delta_k \times P$.  No other productions
  are in~$P''$.  Finally, we let~$F''_q = F_q$ for all~$q \in Q$
  and~$F''_\bot = 0$.  Obviously, $G''$~is eq-restricted, positive, and
  classic.

  In order to better describe the behaviour of~$G''$, let us introduce
  the following notation.  Given a
  tree~$t = \sigma(\seq t1k) \in T_\Sigma$ and a complete left-most
  derivation~$d = (p_1, w_1) \cdots (p_m, w_m)$ of~$G$ for~$t$,
  let~$\seq d1k$ be the derivations for~$\seq t1k$, respectively that
  are incorporated in~$d$ and $h_\sigma = \delta(\seq u1n)$.  Then we
  define the tree~$h(t, d) \in T_{\Delta\cup\Delta\times P}$
  inductively by
  \[ h(t, d) = \langle\delta, p_m \rangle(\seq u1n) \big[h(t_1, d_1),
    \dotsc, h(t_k, d_k) \big] \enspace. \] 
  Using this notation, let us now prove that for each~$q\in Q$ we have
  \begin{equation}
    \label{eq:p1}
    \bigl\{s \in T_{\Delta\cup\Delta\times P} \mid D_{G''}^q(s) \neq
    \emptyset \bigr\} = \bigl\{ h(t, d) \mid t \in \Tsigma, d \in
    D_G^q(t) \bigr\}
  \end{equation}
  and, in turn, every such~$D_{G''}^q(s)$ is a singleton set
  with~$\wt_{G''}(d'') = \wt_G(d)$ for the unique~$d'' \in D_{G''}^q
  \big(h(t, d) \big)$. 

  We start with the inclusion from right to left.  To this end,
  let~$t \in T_\Sigma$ be a tree
  and~$d = (p_1, w_1) \cdots (p_m, w_m)$ be a complete left-most
  derivation of~$G$ for~$t$ to some nonterminal~$q \in Q$.
  Let~$t = \sigma(\seq t1k)$ be the input tree
  with~$h_\sigma = \delta(\seq u1n)$,
  let~$p_m = \sigma(\seq q1k) \to q$ be the production utilized last
  in~$d$, and let~$d_i$ be the complete left-most derivation for~$t_i$
  to~$q_i$ incorporated in~$d$ for every~$i \in [k]$.  For
  every~$i \in [k]$, we utilize the induction hypothesis to conclude
  that~$D_{G''}^{q_i} \bigl(h(t_i, d_i) \bigr)$ is a singleton set, so
  let~$d''_i \in D_{G''}^{q_i} \bigl(h(t_i, d_i) \bigr)$ be the unique
  element, for which we additionally
  have~$\wt_{G''}(d''_i) = \wt_G(d_i)$.  Moreover, for
  every~$i \in [k]$ there is a derivation~$d^\bot_i$ for~$h(t_i, d_i)$
  with weight~$1$ that exclusively utilizes the nonterminal~$\bot$.
  We define
  \[ s = \langle\delta, p_m\rangle(\seq u1n) \bigl[h(t_1, d_1),
    \dotsc, h(t_k, d_k) \bigr] \enspace. \] For every~$i \in [k]$,
  let~$v_i$ be the left-most occurrence of~$x_i$ in~$h_\sigma$.  We
  consider the derivations~$v_1h(t_1, d_1), \dotsc, v_kh(t_k, d_k)$,
  and for every other occurrence~$v$ of~$x_i$ in~$h_\sigma$ we
  consider the derivation~$vd^\bot_i$.  Let~$d''$ be the derivation
  assembled from the considered subderivations followed
  by~$(p''_m, \varepsilon)$, where the production~$p''_m$ at the root
  is $p''_m = \langle\delta, p_m\rangle(\seq u1n) \llbracket \seq q1k
  \rrbracket \stackrel{E, \emptyset}\longrightarrow q$ with the
  constraints~$E = \bigcup_{i=1}^k \pos_{x_i}(h_\sigma)^2$.
  Clearly, the production~$p''_m$ is the only applicable one since the
  only other production whose left-hand side is labeled by
  $\langle \delta, p_m\rangle$ at the root reaches~$\bot \neq q$.
  Reordering the derivation~$d''$ to be left-most, we obtain the
  desired complete left-most derivation~$\underline d''$ for~$s$, for
  which we also have~$\wt_{G''}(\underline d'') = \wt_G(d)$.  This
  proves that~$\underline d''$ is the required single element
  of~$D_{G''}^q(s) = D_{G''}^q \bigl(h(t, d) \bigr) \neq \emptyset$.
  
  On the other hand, consider~$s \in T_{\Delta \cup \Delta \times P}$
  such that there exists a complete left-most
  derivation~$d'' = (p_1'', w''_1) \cdots (p''_m, w''_m)$ for~$s$
  to~$q$; i.e.~$d'' \in D_{G''}^q(s) \neq \emptyset$.  The final
  rule~$p''_m$ that is applied must be of the form
  \[ p''_m = \langle\delta, p\rangle(\seq u1n) \llbracket \seq q1k
    \rrbracket \stackrel{E,\emptyset}\longrightarrow q \]
  with~$\delta(\seq u1n) \llbracket \seq q1k \rrbracket = h_\sigma
  \llbracket \seq q1k \rrbracket$ for some
  symbol~$\sigma \in \Sigma_k$ and
  production~$p = \sigma(\seq q1k) \to q$.  For every~$i \in [k]$, we
  denote by~$w_i$ the unique position
  in~$h_\sigma \llbracket \seq q1k \rrbracket$ labeled by~$q_i$.  By
  the induction hypothesis applied to~$s|_{w_i}$, for which the
  complete left-most derivation~$d''_i$ for~$s|_{w_i}$ to~$q_i$
  incorporated in~$d''$ exists, there exists a tree~$t_i \in \Tsigma$
  and a complete left-most derivation~$d_i$ of~$G$ for~$t_i$ to~$q_i$
  such that~$s|_{w_i} = h(t_i, d_i)$
  and~$\wt_G(d_i) = \wt_{G''}(d''_i)$.  For the
  tree~$t = \sigma (\seq t1k)$ we obtain that~$s = h(t, d)$ for the
  complete left-most derivation~$d \in D_G^q(t)$ given by
  \[ d= (1d_1) \cdots (kd_k) (p,\varepsilon) \enspace, \]
  for which we also have~$\wt_G(d) = \wt_{G''}(d'')$, which completes
  this proof.

  So far, $Q''$~and~$P''$ are larger than $Q$~and~$P$ only by a
  constant (assuming a fixed alphabet~$\Sigma$) caused by the
  additional sink nonterminal~$\bot$ and its productions, but the
  alphabet size increases by the summand~$\abs \Delta \cdot \abs P$.
%
%

  We now delete the annotation with the help of the
  relabeling~$\pi \in \Delta^{\Delta \cup \Delta \times P}$ given
  for every~$\delta \in \Delta$ and~$p \in P$ by~$\pi_\delta =
  \pi_{\langle \delta, p\rangle} = \delta$ following the construction
  in Theorem~\ref{thm:rel}.
  \begin{align*}
    \pi(G'')_u
    &= \sum_{s \in \pi^{-1}(u)} G''_s = \sum_{s \in \pi^{-1}(u)}
      \Bigl( \sum_{q \in Q} F''_q \cdot \wt_{G''}^q(s) \Bigr) =
      \sum_{\substack{q \in Q,\, s \in \pi^{-1}(u) \\ d''
    \in D_{G''}^q(s)}} F''_q \cdot \wt_{G''}(d'') \\
    &\stackrel{\eqref{eq:p1}}= \sum_{\substack{q \in Q,\, s \in
      \pi^{-1}(u) \\ t \in T_\Sigma,\, d \in D_G^q(t) \\ s = h(t, d)}}
    F_q \cdot \wt_G(d) = \sum_{\substack{q \in Q \\ t \in h^{-1}(u)}}
    F_q \cdot \wt_G^q(t) = \sum_{t \in h^{-1}(u)} G_t = h(G)_u 
  \end{align*}
  for every~$u \in T_\Delta$.  The construction of
  Theorem~\ref{thm:rel} is applicable because $\bot$~is clearly a sink
  nonterminal in~$G''$ and $G''$~is an eq-restricted positive classic
  WTGc. \qed
\end{proof}

Let us illustrate the construction on a simple example.

\begin{example}
  Consider the WTA~$G = \bigl(\{q, q'\}, \Sigma, F, P, \mathord{\wt}
  \bigr)$ over the semiring~$\N$ of nonnegative integers with~$\Sigma
  = \{\alpha^{(0)}, \phi^{(1)}, \gamma^{(1)}, \epsilon^{(1)}\}$, $F_q = 0$,
  $F_{q'} = 1$, and the set of productions and their weights given by
  \[ p_1 = \alpha \to_1 q \qquad p_2 = \gamma(q) \to_2 q \qquad  p_3 =
    \epsilon(q) \to_1 q \quad \text{and} \quad  p_4 = \phi(q) \to_1
    q' \enspace. \]
  Then~$\supp(G) = \bigl\{ \phi(t) \mid t \in
  T_{\Sigma\setminus\{\phi\}} \bigr\}$ and~$G_t =
  2^{\abs{\pos_\gamma(t)}}$ for every~$t \in \supp(G)$.  Consider the
  ranked alphabet~$\Delta = \{\alpha^{(0)}, \gamma^{(1)},
  \sigma^{(2)}\}$ and the homomorphism~$h$ induced by $h_\alpha =
  \alpha$, $h_\gamma = h_\epsilon = \gamma(x_1)$, and~$h_\phi = \sigma
  \bigl(\gamma(x_1), x_1 \bigr)$.  Consequently,
  \[ \supp\bigl(h(G) \bigr) = \bigl\{
    \sigma\bigl(\gamma^{n+1}(\alpha), \gamma^n(\alpha) \bigr) \mid n
    \in \N \bigr\} \]
  and $h(G)_t = \sum_{k=0}^{n} \binom{n}{k} 2^k=3^n$ for every~$t =
  \sigma \bigl(\gamma^{n+1}(\alpha), \gamma^n(\alpha) \bigr) \in \supp
  \bigl(h(G) \bigr)$.   A WTGc for~$h(G)$ is constructed as follows. 
  First, we let
  \[ G'' = \bigl(\{q, q', \bot\}, \Delta \cup \Delta \times P, F'',
    P'', \mathord{\wt''} \bigr) \]
  with~$F''_{q'} = 1$, $F''_{q} = F''_\bot = 0$ and the productions
  and their weights are given by
  \begin{align*}
    \langle \alpha, p_1\rangle
    &\to_1 q
    & \langle \gamma, p_2 \rangle(q)
    &\to_2 q 
    & \langle \gamma, p_3 \rangle(q)
    &\to_1 q
    & \langle \sigma, p_4 \rangle\bigl(\gamma(q), \bot \bigr)
    &\stackrel{11=2}\longrightarrow_1 q'
  \end{align*}
  and~$\delta(\bot, \dotsc, \bot) \to_1 \bot$ for all~$\delta \in
  \Delta \cup \Delta \times P$.  Next we remove the second component
  of the symbols of~$\Delta \times P$ and add the weights of
  all productions that yield the same production once the second
  components are removed.  In our example, this applies to the
  production~$\gamma(q) \to q$, which is the result of the two
  productions $\langle \gamma, p_2\rangle(q) \to_2 q$ and $\langle
  \gamma, p_3\rangle(q) \to_1 q$, so its weight is~$2 + 1 = 3$.
  Overall, we obtain the WTGc~$G' = \bigl(\{q, q', \bot\}, \Delta,
  F'', P', \mathord{\wt'} \bigr)$ with the following productions for
  all~$\delta \in \Delta$:
  \begin{align*}
    \alpha
    &\to_1 q
    & \gamma(q)
    &\to_3 q 
    & \sigma\bigl(\gamma(q), \bot\bigr)
    &\stackrel{11=2}\longrightarrow_1 q'
    &\delta(\bot, \dotsc, \bot)
    &\to_1 \bot \enspace. \tag*{\qed}
  \end{align*}
\end{example}

Trees generated by a WTGc must satisfy certain equality constraints on
their subtrees.  Therefore, if we naively swap subtrees of generated
trees, then we might violate such an equality constraint and obtain a
tree that is no longer generated by the WTGc.  Luckily, the particular
kind of WTGc constructed in Theorem~\ref{thm:hom}, namely
eq-restricted positive classic WTGc, allows us to refine
the subtree substitution such that it takes into consideration the
equality constraints in force.  The following definition is the
natural adaptation of~\cite[Definition~5.1]{godoy2013hom} for
(Boolean) tree automata with  
constraints.

\begin{definition}
  \label{df:pumping} 	
  Let~$G = (Q, \Sigma, F, P, \mathord{\wt})$ be an eq-restricted,
  positive, and classic WTGc with sink nonterminal~$\bot$.  Moreover,
  let~$q, q' \in Q$, $t, t' \in T_\Sigma$, and~$d \in D_G^q(t)$ as
  well as~$d' \in D_G^{q'}(t')$ such that~$q \neq \bot \neq q'$ and $d
  = \underline d (p, \varepsilon)$ with the final utilized production
  $p = c[\seq q1k] \stackrel{E, \emptyset}\longrightarrow q \in P$.
  For every~$i \in [k]$ let $w_i = \pos_{x_i}(c)$ and $d_i$~be the
  unique derivation for~$t_i = t|_{\pos_{x_i}(c)}$ incorporated in~$d$.
  Finally, for every tree~$u \in T_\Sigma$ let~$d^\bot_u$ be the
  unique derivation for~$u$ to~$\bot$.  For every~$w \in \pos(t)$,
  for which the derivation for~$t|_w$ incorporated in~$d$ yields~$q'$
  we recursively define the derivation substitution~$d \llbracket d'
  \rrbracket_w$ of~$d'$ into~$d$ at~$w$ and the resulting tree~$t
  \llbracket t' \rrbracket_w^d$ as follows.  If~$w = \varepsilon$,
  then~$d \llbracket d' \rrbracket_{\varepsilon} = d'$ and $t
  \llbracket t' \rrbracket_{\varepsilon}^d = t'$.  Otherwise $w =
  w_j\underline w$ for some~$j \in [k]$ and we have
  \[ d \llbracket d' \rrbracket_w = \word{d'}1k (p, \varepsilon) \qquad
    \text{and} \qquad t \llbracket t' \rrbracket_w^d = c[\seq{t'}1k]
    \enspace, \]
  where for each~$i \in [k]$ we have
  \begin{itemize}
  \item if~$i = j$ (i.e., $w_i$ is a prefix of~$w$),
    then~$d'_i = w_i (d_i \llbracket d' \rrbracket_{\underline w})$
    and $t'_i = t_i \llbracket t' \rrbracket_{\underline w}^{d'_i}$, 
  \item if~$q_i = \bot$ and~$w_i \in [w_j]_{\equiv_E}$
    (i.e., it is a position that is equality restricted to~$w_j$), then
    $d'_i = w_id^\bot_u$ and~$t'_i = u$ with $u = t_j \llbracket t'
    \rrbracket_{\underline w}^{d'_j}$, and 
  \item otherwise $d'_i = w_id_i$ and~$t'_i = t_i$ (i.e., derivation
    and tree remain unchanged).
  \end{itemize}
  It is straightforward to verify that~$d \llbracket d' \rrbracket_w$
  is a complete left-most derivation of~$G$ for~$t \llbracket
  t'\rrbracket_w^d$ to~$q$. \qed
\end{definition}

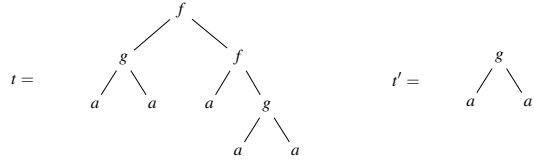
\begin{figure}
  \begin{center}
    \scalebox{.8}{
      \begin{tikzpicture}
        \node at (-2.8, 0) (label) {$t = {}$};
        \node at (0, 0)  (a) {
          \begin{forest}
            for tree={%
              l sep=0.1cm,
              s sep=0.6cm,
              minimum height=0.000008cm,
              minimum width=0.000015cm,
            }
            [$f$
            [$g$
            [$a$][$a$]
            ]
            [$f$
            [$a$][$g$[$a$][$a$]]
            ] 
            ]
          \end{forest}
        };
        \node at (3.5, 0) (label) {$t' = {}$};
        \node at (5, 0)  (a) {
          \begin{forest}
            for tree={%
              l sep=0.1cm,
              s sep=0.6cm,
              minimum height=0.000008cm,
              minimum width=0.000015cm,
            }
            [$g$
            [$a$][$a$]
            ]
          \end{forest}
        };
      \end{tikzpicture} }
  \end{center}
  \caption{Input trees $t$~and~$t'$ from
    Example~\protect{\ref{ex:subst}}.} 
  \label{fig:tree26}
\end{figure}

\begin{example}
  \label{ex:subst}
  We consider the
  WTGc~$G = \big(\{q, \bot\}, \Sigma, F, P, \mathord{\wt}\big)$ with
  input ranked alphabet~$\Sigma = \{a ^{(0)}, g^{(2)}, f^{(2)}\}$,
  final weights $F_q = 1$~and~$F_\bot = 0$ as well as productions
  \[ p_a = a \to_1 q \qquad p_g = g(q, \bot)
    \stackrel{1=2}\longrightarrow_1 q \quad \text{and} \quad p_f = f
    \big(q, f(q, \bot) \big) \stackrel{1=22}\longrightarrow_1 q \]
  besides the sink nonterminal
  productions~$p_\sigma^\bot = \sigma(\bot, \dotsc, \bot) \to_1 \bot$
  for all~$\sigma \in \Sigma$.  As before, for every~$u \in T_\Sigma$
  we let~$d^\bot_u \in D_G^\bot(u)$ be the unique derivation of~$G$
  for~$u$ to~$\bot$, which utilizes only the nonterminal~$\bot$.
  According to Definition~\ref{df:pumping} we choose the
  states~$q = q'$ and the trees $t$~and~$t'$ and derivations
  $d$~and~$d'$ as given in Figure~\ref{fig:tree26} and below.
  \begin{align*}
    d
    &= (p_a, 11) \, (p^\bot_a, 12) \, (p_g, 1) \, (p_a, 21) \,
      (p^\bot_a, 221) \, (p^\bot_a, 222) \, (p^\bot_g, 22) \, (p_f,
      \varepsilon) \\
    d'
    &= (p_a, 1) \, (p^\bot_a, 2) \, (p_g, \varepsilon)
  \end{align*}
  We select that position~$w = 11$ and observe that that the
  derivation for~$t|_{11}$ is~$(p_a, \varepsilon)$, which yields~$q =
  q'$.  We compute~$d \llbracket d' \rrbracket_w$ as follows
  \begin{align*}
    d \llbracket d' \rrbracket_{11}
    &= \Bigl(1 (d'_1 \llbracket d' \rrbracket_1) \Bigr) \, \Bigl(21
      (p_a, \varepsilon) \Bigr) \, \Bigl(22 d^\bot_u \Bigr) \, (p_f,
      \varepsilon) \\
    &= \Biggl(1 \Bigl(1 d' \Bigr) \, \Bigl(2 d_{g(a,a)}^\bot \Bigr) \,
      (p_g, \varepsilon) \Biggr) \, (p_a, 21) \, (22d_u^\bot) \, (p_f,
      \varepsilon) \\ 
    &= (p_a, 111) \, (p^\bot_a, 112) \, (p_g, 11) \,
      (12d_{g(a,a)}^\bot) \, (p_g, 1) \, (p_a, 21) \, (22d_u^\bot) \,
      (p_f, \varepsilon) \enspace,
  \end{align*}
  where $d'_1 = (p_a, 1) \, (p^\bot_a, 2) \, (p_g, \varepsilon)$
  and~$u = g \bigl(g(a, a), g(a, a) \bigr)$.  We note that~$w =
  11$ is explicitly equality constrained to position~$12$ in~$d$ via
  the constraint~$1 = 2$ at position~$1$ and implicitly equality
  constrained to positions $221$~and~$222$ via the constraint~$1 = 22$
  at the root~$\varepsilon$.  Thus, we obtain~$d \llbracket d'
  \rrbracket_{11}$ by substituting~$d'$ into~$d$ at position~$11$
  as well as substituting~$d_{t'}^\bot$ into~$d$ at positions~$12$,
  $221$, and~$222$.  The obtained tree~$t \llbracket t'
  \rrbracket_w^d$ is displayed in Figure~\ref{fig:tree27}.  \qed
\end{example}

\begin{figure}
  \begin{center}
    \scalebox{.8}{
      \begin{tikzpicture}
        \node at (-5, 0) (label) {$t \llbracket t'\rrbracket_{11}^d = {}$};
        \node at (0, 0)  (a) {
          \begin{forest}
            for tree={%
              l sep=0.1cm,
              s sep=0.6cm,
              minimum height=0.000008cm,
              minimum width=0.000015cm,
            }
            [$f$
            [$g$
            [$g$[$a$][$a$]][$g$[$a$][$a$]]
            ]
            [$f$
            [$a$][$g$[$g$[$a$][$a$]][$g$[$a$][$a$]]]
            ] 
            ]
          \end{forest}
        };
      \end{tikzpicture}
    }
  \end{center} 
  \caption{Obtained pumped tree~$t \llbracket t' \rrbracket_{11}^d$ from
    Example~\protect{\ref{ex:subst}}.}
  \label{fig:tree27}
\end{figure}
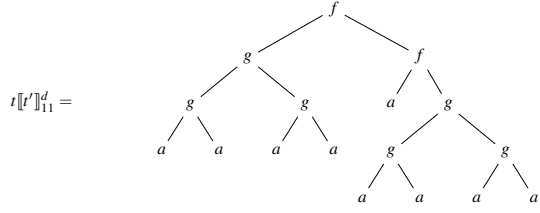

As our example illustrates, the tree~$t \llbracket t' \rrbracket_w^d$
is obtained from~$t$ by (i)~identifying the set of all positions
of~$t$ that are explicitly or implicitly equality constrained to~$w$
by the productions in the derivation~$d$ and (ii)~substituting~$t'$
into~$t$ at every such position.  If~$w' \in \pos(t)$ is parallel to
all positions constrained to~$w$, like position~$21$ in
Example~\ref{ex:subst}, then~$t \llbracket t' \rrbracket_w|_{w'} =
t|_{w'}$.  Note that~$t|_{21}$ is equal to the replaced
subtree~$t|_{11}$, but we only replace constrained subtrees and not
all equal subtrees.

This substitution allows us to prove a pumping lemma for
eq-restricted, positive, and classic WTGc, which can generate all
(nondeleting and nonerasing) homomorphic images of regular weighted
tree languages by Theorem~\ref{thm:hom}.  To this end, we need some
final notions.  Let~$G = (Q, \Sigma, F, P, \mathord{\wt})$ be a WTGc.
Moreover, let~$p = \ell \stackrel{E,D}\longrightarrow q\in P$ be a
production.  We define the \emph{height}~$\height(p)$ of~$p$
by~$\height(p) = \height(\ell)$ (i.e., the height of its left-hand
side).  Moreover, we let 
\[ \height(P) = \max \bigl\{\height(p) \mid p \in P \bigr\} \qquad
  \text{and} \qquad \height(G) = (\abs Q + 1) \cdot \height(P)
  \enspace. \]

\begin{lemma}
  \label{lm:existence of subs}
  Let~$G = (Q, \Sigma, F, P, \mathord{\wt})$ be an eq-restricted,
  positive, and classic WTGc with sink nonterminal~$\bot$.  There
  exists~$n \in \nat$ such that for every tree~$t_0 \in T_\Sigma$,
  nonterminal~$q \in Q \setminus \{\bot\}$, and derivation~$d \in
  D_G^q(t_0)$ such that~$\height(t_0) > n$ and~$\wt_G(d) \neq
  0$ there are infinitely many trees~$t_1, t_2, \dotsc$ and
  derivations~$d_1, d_2, \dotsc$ such that~$d_i \in D_G^q(t_i)$
  and~$\wt_G(d_i) \neq 0$ for all~$i \in \N$.
\end{lemma}

\begin{proof}
  Without loss of generality, suppose that for
  every~$c[\seq q1k] \stackrel{E,\emptyset}\longrightarrow q' \in P$
  with~$q' \neq \bot$ and~$k \neq 0$ there exists~$i \in [k]$ such
  that~$q_i \neq \bot$.  This can easily be achieved by introducing a
  copy~$\top$ of nonterminal~$\bot$ and replacing one instance
  of~$\bot$ by~$\top$ in offending productions.  Similarly, we can
  assume without loss of generality that the construction in the proof
  of Lemma~\ref{lm:zero} has been applied to~$G$.  If this is the
  case, then we can select~$n = \height(G)$.  Let~$t_0 \in T_\Sigma$
  be such that~$\height(t_0) > n$.  Let~$Q' = Q \setminus \{\bot\}$,
  $d \in D_G^q(t_0)$ be a derivation with~$\wt_G(d) \neq 0$, and
  select a position~$w \in \pos(t_0)$ of maximal length such that $d$
  incorporates a derivation for~$t_0|_w$ to some~$q' \in Q'$.  Then
  \[ \abs w \geq \height(t_0) - \height(P) \geq \height(G) -
    \height(P) = \abs Q \cdot \height(P) \enspace, \]
  which yields that at least~$\abs Q$ proper prefixes~$w'$ of~$w$
  exist such that~$d$ incorporates a derivation for~$t_0|_{w'}$ to
  some~$q' \in Q'$.  Hence there exist prefixes~$w', w''$ of~$w$ such
  that $d$~incorporates a derivation~$d'$ for~$t' = t_0|_{w'}$ to~$q' \in
  Q'$ as well as a derivation for~$t_0|_{w''}$ to the same
  nonterminal~$q'$.  Then~$d \llbracket d' \rrbracket_{w''}$ is a
  derivation of~$G$ for~$t_1 = t \llbracket t' \rrbracket_{w''}^d$
  to~$q$ with~$\height(t_1) > \height(t_0)$.  Since we
  achieve the same state~$q$, the annotation of the proof of
  Lemma~\ref{lm:zero} guarantees that~$\wt_G(d_1) \neq 0$.  Iterating
  this substitution yields the desired trees~$t_1, t_2, \dotsc$ and
  derivations~$d_1, d_2, \dotsc$.  \qed 
\end{proof}

A WTGc generating a (nondeleting and nonerasing) homomorphic image of
a regular weighted tree language, if constructed as described in
Theorem~\ref{thm:hom}, will never have overlapping constraints since
constraints always point to leaves of the left-hand sides of
productions as required by classic WTGc.  It is intuitive that this
limitation to the operating range of constraints leads to an actual
restriction in the expressive power of WTGc, but we will only prove it
for eq-restricted, positive, and classic WTGc.

\begin{proposition}
  Let~$\mathbb{S}$ be a zero-sum free semiring.  The class of
  positive constraint-regular weighted tree languages
  is strictly more expressive than the class of weighted tree
  languages generated by eq-restricted, positive, and classic WTGc.
\end{proposition}

\begin{proof}
  Let us consider the positive WTGc~$G = \bigl(\{q, q'\}, \Sigma, F,
  P, \mathord{\wt}\bigr)$ with input ranked alphabet~$\Sigma =
  \{f^{(2)}, \underline f^{(2)}, g^{(2)}, a^{(0)}\}$, final weights
  $F_q = 1$~and~$F_{q'} = 0$, and the following 
  productions, of which each has weight~$1$.
  \begin{align*}
    a \to_1 q' \qquad \qquad
    g(q', q') \to_1 q \qquad \qquad
    f(q, q) \stackrel{12=21}\longrightarrow_1 q \qquad \qquad
    \underline f(q, q)\stackrel{12=21}\longrightarrow_1 q
  \end{align*}
  The first two productions are only used on leaves and on subtrees of
  the form~$g(a, a)$.  Every other position~$w$ (i.e., neither leaf
  nor position with two leaves as children) is labeled either~$f$
  or~$\underline f$ and additionally every derivation enforces the
  constraint~$12 = 21$, so the subtrees $t|_{w12}$~and~$t|_{w21}$ of
  the input tree~$t$ need to be equal for a complete derivation of~$G$
  to exist.

  For the sake of a contradiction, suppose that an eq-restricted,
  positive, and classic
  WTGc~$G' = (Q', \Sigma, F', P', \mathord{\wt'})$ exists that is
  equivalent to~$G$.  We recursively define the trees
  $t_n \in T_\Sigma$~and~$t'_n \in T_\Sigma$ for every~$n \in \nat$
  with~$n \geq 1$ by
  \begin{align*}
    t_0
    &= a \qquad
    &\qquad t_1
    &= g(t_0, t_0) \qquad
    &\qquad t_{n+1}
    &= f(t_n, t_n) \\
    t'_0
    &= a \qquad
    &\qquad t'_1
    &= g(t'_0, t_0) \qquad
    &\qquad t'_{n+1}
    &= \underline f(t'_n, t_n)
  \end{align*}
  Clearly, $t_n$~and~$t'_n$ are both complete binary trees of
  height~$n$.  Naturally, the leaves are labeled~$a$, and the
  penultimate level in both trees is always labeled~$g$.  In~$t_n$ the
  remaining levels are universally labeled~$f$, whereas in~$t'_n$ the
  left-most spine on those levels is labeled~$\underline f$.  We
  illustrate an example tree~$t'_n$ in Figure~\ref{fig:t1}.
  Obviously~$G(t_n) = 1$~as well as~$G(t'_n) = 1$ for
  every~$n \in \nat$ with~$n \geq 1$.  Furthermore we note that the
  derivations of~$G$ only enforce equality constraints on positions of
  the form $w12$~or~$w21$, but
  since~$\pos_{\underline f}(t'_n) \subseteq\{1\}^*$, the positions,
  in which the labels in~$t_n$~and~$t'_n$ differ, are not affected by
  any equality constraint.  This can be used to verify that
  $G(t'_n) = 1$ for each~$n \geq 1$.

  In the following, let $n = 3 \height(G') + 2$.  Since $G'$~is
  equivalent to~$G$, we need to have~$G'(t'_n) = 1$ as well, which
  requires a complete derivation of~$G'$ for~$t'_n$ to some final
  nonterminal~$q_0 \in Q'$.  Let~$d \in D_{G'}^{q_0}(t'_n)$ be such a
  derivation.  Moreover, let~$d = \underline d (p, \varepsilon)$ for
  some production~$p = c[\seq q1k] \stackrel{E,
    \emptyset}\longrightarrow q_0 \in P'$.  Since the input
  tree~$t'_n$ contains positions
  \[ \Bigl\{1^i = \underbrace{11\dotsm 1}_{i \text{ times}} \mid 0
    \leq i \leq n \Bigr\} \subseteq \pos(t'_n) \enspace, \] there must
  exist~$j \in \N$ such that~$c(1^j) = x_1$; i.e., position~$1^j$ is
  labeled~$x_1$ in~$c$.  Obviously, $j \leq \height(G')$, so the
  height of the subtree~$t'' = t'_n|_{1^j}$, which is still a complete
  binary tree, is at least~$2 \height(G') + 2$.  We can thus apply
  Lemma~\ref{lm:existence of subs} to the tree~$t''$ in such a way
  that it modifies its second direct subtree (starting
  from~$1^j \in \pos(t'_n)$, we descend to~$1^j2$; from there, we
  either find a subderivation to some nonterminal different
  from~$\bot$, or all subtrees below~$1^j2$ are copies of subtrees
  below~$1^j1$, and in that case, we apply the pumping to an equality
  constrained subtree below~$1^j1$, which then also modifies the
  corresponding subtree below~$1^j2$).  Let~$u$ be the such obtained
  pumped tree, which according to zero-sum freeness and
  Lemma~\ref{lm:existence of subs} is also in the support of~$G'$;
  i.e., $u \in \supp(G')$.  Let~$d'$ be the derivation constructed in
  Lemma~\ref{lm:existence of subs} corresponding to~$u$.  We
  have~$u(1^{j-1}) = \underline f$, so the position~$1^{j-1}$ is
  labeled~$\underline f$.  Since $G$~and~$G'$ are equivalent, there
  must be a derivation of~$G$ for~$u$ as well, which enforces the
  equality constraint $u|_{1^{j-1}12} = u|_{1^{j-1}21}$.  By
  construction we have~$t'_n|_{1^{j-1}12} \neq u|_{1^{j-1}12}$.  Since
  the positions $1^{j-1}12$~and~$1^{j-1}21$ have no common suffix,
  this equality can only be guaranteed by~$G'$ if
  $1^{j-1}12$~and~$1^{j-1}21$ are themselves (explicitly or
  implicitly) equality constrained in~$d'$. The potentially several
  constraints that achieve this must of course be located at prefixes
  of $1^{j-1}12$~and~$1^{j-1}21$, and since the production used
  in~$d'$ at the root is still~$p$ and stretches all the way to~$1^j$,
  this can only be achieved if~$d'$ enforces~$1^{j-1}1 = 1^{j-1}2$
  via~$p$ at the root as well as~$1 = 2$ at~$1^{j-1}1$ or
  at~$1^{j-1}2$.  However, this is a contradiction as
  $u(1^{j-1}1) = \underline f \neq f = u(1^{j-1}2)$, so we cannot have
  an explicit or implicit equality constraint between
  $1^{j-1}12$~and~$1^{j-1}21$, so~$u|_{1^{j-1}21} =
  t'_n|_{1^{j-1}21}$, but contradicts that~$G$ has a complete
  derivation for~$u$. \qed 
\end{proof}

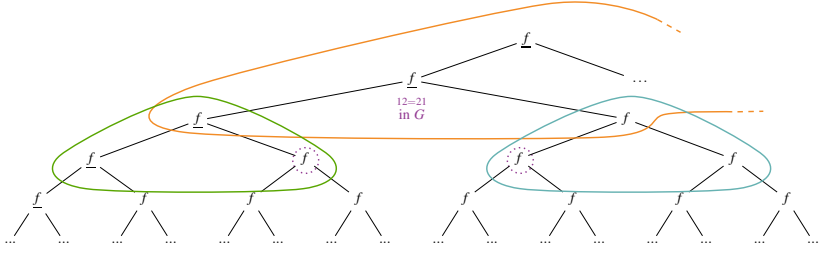
\begin{figure}\label{fig:separate}
  \begin{center}
    \scalebox{.68}{
      \begin{tikzpicture}
        \node at (0,.7) (c){\textcolor{violet!70!lightgray}
          {$\scriptstyle 12 = 21$}};
        \node at (0,0.45) (c)
        {\begin{small}\textcolor{violet!70!lightgray}{in $G$}\end{small}};
        \path (2.1, -0.45) node[draw,thick,dotted,violet!70!lightgray,fill=white!0,minimum size=0.5cm,circle](y){} ;
        \path (-2.05, -0.45) node[draw,thick,dotted,violet!70!lightgray,fill=white!0,minimum size=0.5cm,circle](z){} ;
        \node at (0, 0)  (a) {
          \begin{forest}
            for tree={%
              l sep=0.1cm,
              s sep=0.6cm,
              minimum height=0.000008cm,
              minimum width=0.000015cm,
            }
            [$\underline f$
            [$\underline f$
            [$\underline f$
            [$\underline f$
            [$\underline f$
            [$...$]
            [$...$]
            ]
            [$f$ [$...$] [$...$]
            ]
            ]
            [$f$ [$f$ [$...$] [$...$]
            ]
            [$f$ [$...$] [$...$]
            ]
            ]
            ]
            [$f$
            [$f$ [$f$ [$...$] [$...$]
            ]
            [$f$ [$...$] [$...$]
            ]
            ]
            [$f$ [$f$ [$...$] [$...$]
            ]
            [$f$ [$...$] [$...$]
            ]
            ]
            ] 
            ]
            [$\qquad \qquad \qquad \qquad \qquad \dots \qquad \qquad
            \qquad \qquad \qquad$] 
            ]
          \end{forest}
        }; 
        \draw [thick, color=lightgray!20!orange] plot [mark=none, smooth] coordinates {(4.75,2.3) (2.6,2.6) (-4.2,1) (-4.4,0.1) (3.6,0) (4.9,0.45) (6.2,0.5)};
        \draw [thick, dashed, color=lightgray!20!orange] plot [mark=none, smooth] coordinates {(4.75,2.3) (5.2,2.05)};
        \draw [thick, dashed, color=lightgray!20!orange] plot [mark=none, smooth] coordinates {(6.3,0.5) (6.8,0.51)};
        \draw [thick, color=green!30!olive] plot [mark=none, smooth cycle] coordinates {(-2.1, -1) (-1.6, -0.35) (-4.15, 0.8) (-6.8, -0.35) (-6.3, -1)};
        \draw [thick, color=teal!60!white] plot [mark=none, smooth cycle] coordinates {(6.3,-1) (6.8, -0.35) (4.25, 0.8) (1.6, -0.35) (2.1,-1)};
      \end{tikzpicture} } 
  \end{center}
  \caption{A snippet of the tree~$t'_n$ and the productions used by~$G'$.}
  \label{fig:t1}
\end{figure}

Although for zero-sum free semirings, the support of a regular
weighted tree language is again regular, in general, the converse is
not true,
so we cannot apply the decision procedure of~\cite{godoy2013hom} to
the support of a homomorphic image in order to decide its regularity.
Instead, we hope to extend the unweighted argument in a way that
tracks the weights sufficiently close.  For this, we prepare two
decidability results, which rely mostly on the corresponding results
in the unweighted case.  To this end, we need to relate our WTGc
constructed in Theorem~\ref{thm:hom} to the classic TGc used
in~\cite{godoy2013hom}.  At this point we mention that their classic
TGc additionally require that equality constrained positions have the
same nonterminal label.  Compared to our eq-restriction this change is
entirely immaterial in the unweighted case.

\begin{theorem}
  \label{thm:decid}
  Let $\mathbb S$ be a zero-sum free semiring.  Moreover, let $G = (Q,
  \Sigma, F, P, \mathord{\wt})$ be a WTA and~$h \in
  T_\Delta^{\Tsigma}$~be a nondeleting and nonerasing tree
  homomorphism.  Finally, let~$G' = h(G)$.  Emptiness and finiteness
  of~$\supp(G')$ are decidable. 
\end{theorem}

\begin{proof}
  We apply the construction in the proof of Lemma~\ref{lm:zero} to the
  eq-restricted, positive, and classic WTGc~$G' = (Q', \Sigma, F', P',
  \mathord{\wt'})$ constructing according to Theorem~\ref{thm:hom}.
  In this manner we ensure that all derivations have non-zero weight.
  Due to zero-sum freeness, we can now simply drop the weights and
  obtain a eq-restricted, positive, and classic TGc~$G'' = (Q'',
  \Sigma, F'', P'')$ generating~$\supp(G')$.  Emptiness and finiteness
  are decidable for the tree language~$\supp(G')$ generated by~$G''$
  according to~\cite[Corollaries~5.11 \& 5.20]{godoy2013hom}.  \qed
\end{proof}


\section*{Conflict of interest}
The authors declare that they have no conflict of interest.



\end{document}